\documentclass[11pt]{article}
\usepackage{amssymb,amsthm,amsmath}
\usepackage{graphicx}
\usepackage{fullpage}
\usepackage{subfigure}
\usepackage{graphics}
\usepackage{hyperref}

\newtheorem{theorem}{Theorem}[section]
\newtheorem{lemma}[theorem]{Lemma}
\newtheorem{proposition}[theorem]{Proposition}

\newtheorem{definition}[theorem]{Definition}

\newtheorem{conjecture}{Conjecture}

\newcommand{\diam}{\mathsf{diam}}
\newcommand{\depth}{\mathsf{depth}}

\newcommand{\funnels}{\mathsf{Funnels}}
\newcommand{\pyramids}{\mathsf{Pyramids}}
\newcommand{\len}{\mathsf{len}}
\newcommand{\ball}{\mathsf{ball}}
\newcommand{\dil}{\mathsf{dil}}
\newcommand{\distortion}{\mathsf{distortion}}
\newcommand{\Lip}{\mathsf{Lip}}
\newcommand{\nca}{\mathsf{nca}}

\newcommand{\dem}{\mathsf{dem}}
\newcommand{\capa}{\mathsf{cap}}
\newcommand{\1}{\mathbf{1}}
\newcommand{\Ind}{\mathbf{1}}

\newcommand{\etal}{\emph{et al.}}

\renewcommand{\phi}{\varphi}
\newcommand{\eps}{\varepsilon}

\title{Non-positive curvature, and the planar embedding conjecture}

\author{Anastasios Sidiropoulos\thanks{Department of Computer Science, University of Illinois at Urbana-Champaign; \url{http://www.sidiropoulos.org/}; \url{sidiropo@illinois.edu}.  Supported in part by David and Lucille Packard Fellowship and by NSF grants CCF-0915984 and CCF-0915519.}}

\begin{document}

\setcounter{page}{0}
\maketitle

\begin{abstract}
The planar embedding conjecture asserts that any planar metric admits an embedding into $L_1$ with constant distortion.
This is a well-known open problem with important algorithmic implications, and has received a lot of attention over the past two decades.
Despite significant efforts, it has been verified only for some very restricted cases, while the general problem remains elusive.

In this paper we make progress towards resolving this conjecture.
We show that every planar metric of non-positive curvature admits a constant-distortion embedding into $L_1$.
This confirms the planar embedding conjecture for the case of non-positively curved metrics.

%Our result implies in particular a reformulation of the general problem as a certain ``hyperbolization'' criterion: The planar embedding conjecture is true if and only if every planar metric admits a constant-distortion embedding into a convex combination of planar metrics of non-positive curvature.
\end{abstract}

\thispagestyle{empty}
\newpage

\section{Introduction}
\label{sec:intro}

If $(X,d_X),(Y,d_Y)$ are metric spaces,
and $f : X \to Y$ is injective, 
the \emph{distortion} of $f$ is defined to be $\distortion(f) = \|f\|_\Lip \cdot \|f^{-1}\|_\Lip$, where
$\|f\|_\Lip = \sup_{x \neq y \in X} \frac{d_Y(f(x),f(y))}{d_X(x,y)}$.
For any metric space $(X,d)$, we use $c_1(X,d)$ to denote
the {\em $L_1$ distortion of $(X,d)$}, i.e.
the infimum over all numbers $D$ such that
$X$ admits an embedding into $L_1$ with distortion $D$.
For a graph $G=(V,E)$ we write
$c_1(G) = \sup c_1(V, d)$ where $d$ ranges over all shortest-path metrics supported on $G$,
and for a family $\mathcal F$ of graphs, we write $c_1(\mathcal F) = \sup_{G \in \mathcal F} c_1(G)$.
Thus for a family $\mathcal F$ of finite graphs, $c_1(\mathcal F) \leq D$ if and only if
every geometry supported on a graph in $\mathcal F$ embeds into $L_1$
with distortion at most $D$.

In the seminal works of Linial-London-Rabinovich \cite{LLR95}, and later
Aumann-Rabani \cite{AR98} and Gupta-Newman-Rabinovich-Sinclair \cite{GNRS99},
the geometry of graphs is related to the classical
study of the relationship between flows and cuts.

A {\em multi-commodity flow instance in $G$}
 is specified
by a pair of non-negative mappings $\capa : E \to \mathbb R$ and $\dem : V \times V \to \mathbb R$.
We write $\mathsf{maxflow}(G; \capa, \dem)$ for the value of the {\em maximum concurrent flow} in this instance,
which is the maximal value $\varepsilon$
such that $\varepsilon \cdot \dem(u,v)$ can be simultaneously routed between every pair $u,v \in V$
while not violating the given edge capacities.

A natural upper bound on $\mathsf{maxflow}(G; \capa, \dem)$ is given by the {\em sparsity}
of any cut $S \subseteq V$:
\begin{equation}\label{eq:sparse}
\frac{\sum_{uv \in E} \capa(u,v) |\1_S(u)-\1_S(v)|}{\sum_{u,v \in V} \dem(u,v) |\1_S(u)-\1_S(v)|},
\end{equation}
where $\1_S : V \to \{0,1\}$ is the indicator function for membership in $S$.
We write $\mathsf{gap}(G)$ for the maximum gap between the value of the flow
and the upper bounds given by \eqref{eq:sparse}, over all multi-commodity flow instances on $G$.
This is the {\em multi-commodity max-flow/min-cut gap for $G$}.
The fundamental connection between embeddings into $L_1$ and multi-commodity flows is captured in the following result.

\begin{theorem}[\cite{LLR95, GNRS99}]\label{thm:gap}
For every graph $G$, $c_1(G) = \mathsf{gap}(G)$.
\end{theorem}

In particular, combined with the techniques of \cite{LR99, LLR95}, this implies that for any graph $G$, 
there exists a $c_1(G)$-approximation for the general Sparsest Cut problem.

\subsection{The planar embedding conjecture}
It has been shown by \cite{LLR95,LR99} that for general graphs, $c_1(G) = \Omega(\log n)$, and there has since been a lot of effort in trying to prove that $c_1(G)$ is bounded by some universal constant for interesting classes of graphs.
The most well-known open case is the so-called \emph{planar embedding conjecture}, summarized in the following:

\begin{conjecture}[Planar embedding conjecture]
For every planar graph $G$, $c_1(G) = O(1)$.
\end{conjecture}

Despite several attempts on resolving this question, there has only been very little progress.
More specifically, the work of Okamura \& Seymour \cite{Okamura198175} implies that the metric induced on a \emph{single face} of a planar graph embeds with constant distortion into $L_1$.
In \cite{GNRS99} it is shown that $c_1(G)=O(1)$ for any series-parallel, or outerplanar graph $G$.
This result was extended to $O(1)$-outerplanar graphs in \cite{CGNRS06}.
Chakrabarti \etal~\cite{DBLP:conf/focs/ChakrabartiJLV08} obtained constant distortion embeddings of graphs that exclude a $(K_5 \setminus e)$-minor.
Note that even the case of planar graphs of treewidth $3$ remains open.
We remark that the best-known upper bound on $c_1(G)$ for planar graphs is $O(\sqrt{\log n})$, due to Rao \cite{DBLP:conf/compgeom/Rao99}, while the best-known lower bound is 2, due to Lee \& Raghavendra \cite{DBLP:journals/dcg/LeeR10}.

\subsection{Generalizations: The GNRS conjecture}

Gupta, Newman, Rabinovich, and Sinclair \cite{GNRS99} posed the following generalization of the planar embedding conjecture, which seeks to 
{\em characterize} the graph families $\mathcal F$ such that $c_1(\mathcal F) = O(1)$, which by Theorem \ref{thm:gap} also characterizes all graphs with multi-commodity gap bounded by some universal constant:

\begin{conjecture}[GNRS conjecture \cite{GNRS99}]
\label{conj:gnrs}
For every family of finite graphs $\mathcal F$, one has $c_1(\mathcal F) = O(1)$ if and only if $\mathcal F$ forbids some minor.
\end{conjecture}

We note that a strengthening of the GNRS conjecture for \emph{integral} multi-commodity flows has also been considered \cite{DBLP:journals/jct/ChekuriSW13}. This is a seemingly harder problem, and progress has been even more limited in this case.

At first glance, it might appear that the GNRS conjecture is a vast generalization of the planar embedding conjecture, since planar graphs exclude $K_5$ as a minor.
Despite this, Lee \& Sidiropoulos \cite{DBLP:conf/stoc/LeeS09} have shown that the GNRS conjecture is \emph{equivalent} to the conjunction of the planar embedding conjecture, with 
the manifestly simpler \emph{$k$-sum embedding conjecture} summarized bellow.
For a graph family ${\cal F}$, let $\oplus_k {\cal F}$ denote the closure of ${\cal F}$ under $k$-clique sums (see \cite{DBLP:conf/stoc/LeeS09} for a more detailed exposition).
We note that the case $k=1$ is folklore, while recently progress has been reported for the case $k=2$ by Lee and Poore \cite{2sums}; even for $k=2$ however, the problem remains open.

\begin{conjecture}[$k$-sum conjecture \cite{DBLP:conf/stoc/LeeS09}]
\label{conj:ksum}
For any family of graphs $\mathcal F$, we have $c_1(\mathcal F) = O(1)$ if and only if $c_1(\oplus_k \mathcal F) = O(1)$
for every $k \in \mathbb N$.
\end{conjecture}

%Roughly speaking, the $k$-sum conjecture captures the special case of the GNRS conjecture for graphs of bounded treewidth.
It is therefore apparent that the planar embedding conjecture is a major step towards determining the multi-commodity gap in \emph{arbitrary} graphs.

\subsection{Our results}

All previous attempts on the planar embedding conjecture have been \emph{topological} in nature, meaning that they seek to obtain constant-distortion embeddings by restricting the topology of the planar graph.
As a consequence, all known methods are insufficient even for planar graphs of treewidth 3.
%As a consequence, all previous methods are restricted to graphs of constant treewidth.

We depart from this paradigm by instead restricting the \emph{geometry} of the planar metric.
For any metric $(X,d)$, we have that $(X,d)$ is the shortest-path metric of a planar graph if and only if it can be realized as a set of points in a simply-connected (i.e.~planar) surface.
We say that a planar metric is \emph{non-positively curved} if it can be realized as a set of points in a surface of non-positive curvature (see Section \ref{sec:prelim} for the definition of non-positively curved spaces).
This leads to a natural, and very rich class of planar metrics.
For instance, non-positively curved planar metrics include all trees, all regular grids (up to constant distortion), and arbitrary subsets of the hyperbolic plane $\mathbb{H}^2$.

\begin{center}
\scalebox{0.7}{\includegraphics{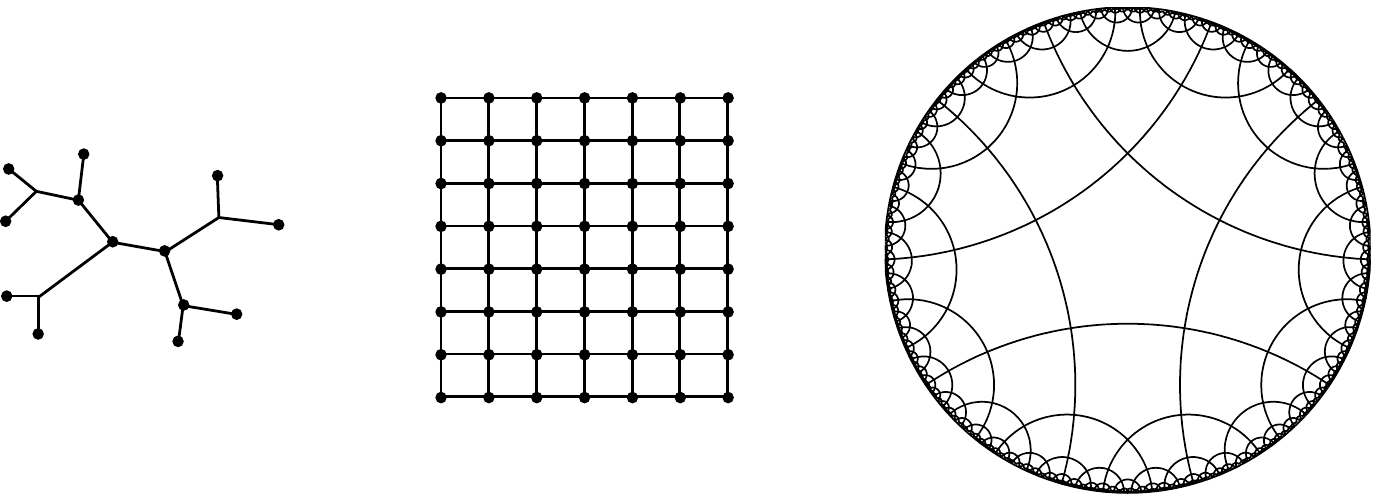}}
\end{center}

Our main result is as follows.

\begin{theorem}[Main]\label{thm:main}
There exists a universal constant $\gamma > 1$, such that 
every non-positively curved planar metric admits an embedding into $L_1$ with distortion at most $\gamma$.
\end{theorem}

Since we are motivated by the applications of metric embeddings in computer science, we will restrict our discussion to finite metrics.
We remark however that our result can be extended to obtain constant-distortion embeddings of arbitrary simply-connected surfaces of non-positive curvature into $L_1$.\footnote{This connection was pointed out  by James R.~Lee.}

We note that embeddings of various hyperbolic spaces have been previously considered.
We refer to \cite{DBLP:conf/focs/KrauthgamerL06,Bonk_embeddingsof,hyperbolic_into_products_of_trees,Buyalo:2005to}.
However, none of the previous results captures $L_1$ embeddings of arbitrary non-positively curved planar metrics.
In fact, our approach is significantly different than all previous works.

%Interestingly, our result also implies a re-statement of the planar embedding conjecture.

\subsection{A high-level overview of our approach}
We now give an informal, and somewhat imprecise overview of some of the main challenges that we face when trying to embed non-positively curved planar metrics into $L_1$.

\paragraph{Distributions over monotone cuts.}
Let $(X,d)$ be a metric space.
We will use the standard representation of $L_1$ as the cone of cut pseudo-metrics (see Section \ref{sec:prelim} for the definition).
This means that in order to embed a space into $L_1$ with constant distortion, it suffices to find a probability distribution over cuts, such that the probability that any pair of points $x,y$ gets separated, is $\Theta(\alpha \cdot d(x,y))$, for some normalization factor $\alpha>0$.

It follows by the work of Lee and Raghavendra \cite{DBLP:journals/dcg/LeeR10} (see also \cite{DBLP:conf/focs/ChakrabartiJLV08}) that when seeking a constant-distortion embedding of certain spaces into $L_1$ it suffices to consider distributions over a specific type of cuts, called \emph{monotone}.
More precisely, let $x$ be a fixed point.
We say that a cut $S$ is monotone (w.r.t.~$x$) if every shortest path starting from $x$ crosses $S$ at most once.
Let us say that a metric space is a \emph{bundle} if there exist two points $s,t$, such that for every point $z$, there exists an $s$-$t$ geodesic containing $z$.
Then it is shown in \cite{DBLP:journals/dcg/LeeR10} that a bundle admits a constant-distortion embedding into $L_1$ if and only if it is a convex combination of monotone cuts (i.e.~a convex combination of cut pseudo-metrics, where every indicator set is a monotone cut). 
It is easy to show that every finite non-positively curved metric admits an isometric embedding into a bundle.
We can therefore focus our efforts into finding a good distribution over monotone cuts.

\paragraph{The structure of monotone cuts in non-positively curved spaces.}
It is convenient to demonstrate the main ideas using the following example of a ``pinched square''.
Let $X=[0,1]^2$, endowed with the Euclidean distance.
The space $X$ can be embedded isometrically into $L_1$ by taking an appropriate distribution over random half-plane cuts (e.g.~by choosing a uniformly random point $p\in X$, and taking the half-plane supported by a line passing through $p$ forming a uniformly random angle with the $x$-axis).
Let $A$ be one of the sides of $X$, and let $Y=X/A$ be the quotient space obtained by contracting $A$ into a single point, which we will refer to as the \emph{basepoint}.

\begin{center}
\scalebox{0.6}{\includegraphics{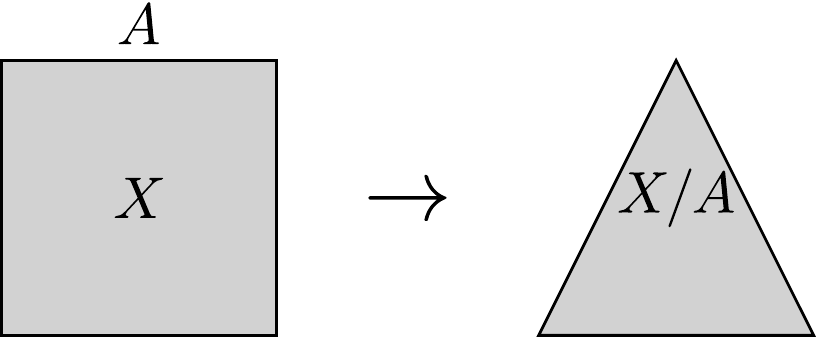}}
\end{center}

Strictly speaking, the resulting $Y$ is not a space of non-positive curvature (in particular, there exist pairs of points in $Y$ with two distinct geodesics joining them).
However, $Y$ admits a constant-distortion embedding into a planar surface of non-positive curvature, so in order to simplify the exposition, we may use $Y$ without loss of generality.

It is fairly easy to see that even though $Y$ might ``look'' like a triangle, its geometry is far from that of a flat Euclidean triangle.
In fact, one can show that $Y$ cannot be embedded into the Euclidean plane with bounded distortion.
As a consequence, embedding $Y$ into $L_1$ requires a significantly more involved distribution over cuts.
Such a distribution can be constructed using cuts of the following form:
For every $r\in [0,1]$, we have a family of cuts $S$ that are contained inside the ball of radius $r$ from the basepoint, and with boundary $\partial C$ given by a function of period $\Theta(r)$.
Roughly speaking, these cuts can be obtained by random shifts along the $x$-axis of cuts from the following infinite family:

\begin{center}
\scalebox{1.2}{\includegraphics{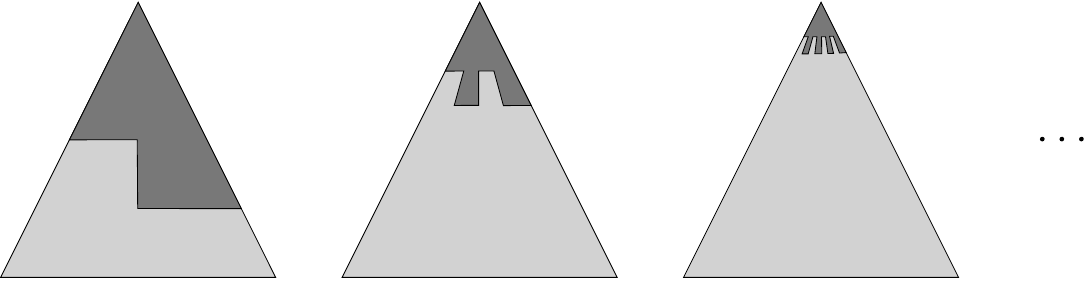}}
\end{center}

Here, the probability of a cut decreases when $r\to 0$.
It is important to note that the structure of a cut depends on the distance of its boundary to the basepoint.
It can be shown that this is the case for \emph{any} constant distortion embedding of $Y$ into $L_1$.
Moreover, in an constant-distortion embedding, this transition has to happen in a smooth way as $r\to 0$.

\paragraph{Handling multiple scales.}

Suppose now that we modify the space $Y$ as follows.
Let $R$ be a ray in $Y$, i.e.~an unbounded geodesic starting at the basepoint, and 
let $R'$ be a suffix of $R$.
Cutting $Y$ along $R'$ introduces two copies $R_1'$, $R_2'$ of $R'$ as segments of the boundary.
We glue a copy of $Y$ along $R_1'\cup R_2'$ as follows:
\begin{center}
\scalebox{0.6}{\includegraphics{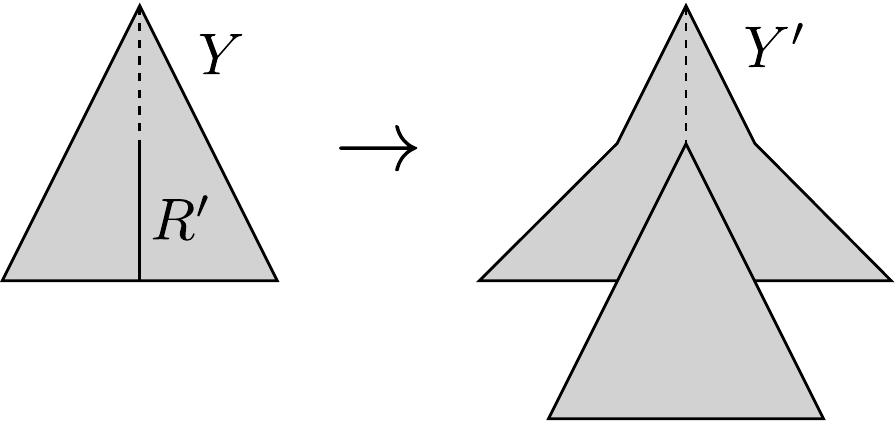}}
\end{center}
The resulting space $Y'$ again embeds with constant distortion into a planar metric of non-positive curvature.
Constructing a constant-distortion embedding for $Y'$ requires the use of even more intricate families of cuts.
Intuitively, a single cut now has to ``gracefully'' combine information form multiple different scales.
Let $p_1,p_2\in Y'$ be the basepoints of the two copies of $Y$ in $Y'$.
The structure of  a ``typical'' cut $S$ has to depend on the distances between $\partial S$, and \emph{both} $p_1$, and $p_2$:
\begin{center}
\scalebox{1}{\includegraphics{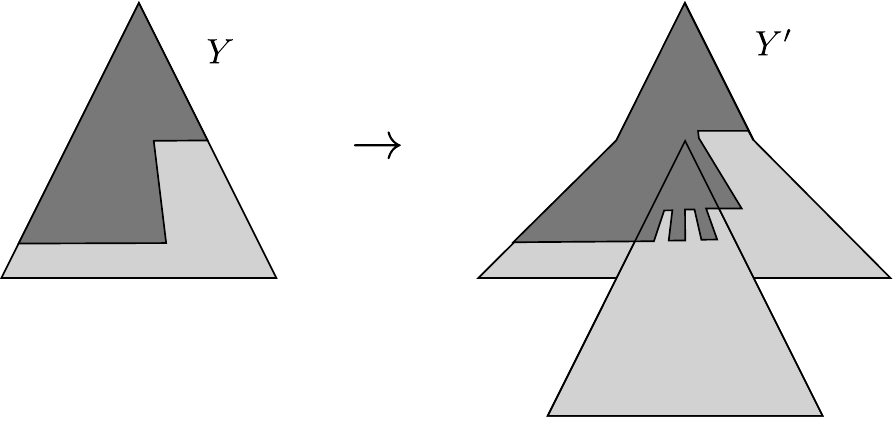}}
\end{center}
A naive way to address this problem would be to define a distribution over cuts for every single scale, and then try to combine them into a single distribution.
The problem with this approach is that cuts from different scales (in our example, cuts for the two different copies of $Y$) might not agree on their boundary.
This disagreement results in larger distortion every time we combine two different scales.
Since there can be many scales, this methods leads to unbounded distortion.

We overcome this obstacle by designing a distribution over cuts that is \emph{scale-independent}.
This is done by starting with a distribution over cuts that handles large distances, and gradually modifying it to handle smaller scales.
The main technical contribution of this paper is showing that  in a non-positively curved surface, this can be done without increasing the distortion.

\subsection{Preliminaries}\label{sec:prelim}

We now review some basic definitions and notions which
appear throughout the paper.

\paragraph{Graphs.}
Let $G$, and let $S\subseteq V(G)$.
We denote by $G[S]$ the subgraphs of $G$ induced by $S$,
i.e. $G[S] = \left(S, E(G) \cap {S \choose 2}\right)$.
We will consider graphs with every edge having a non-negative length.
We say that a graph is \emph{unweighted} if all of its edges have unit length.
Let $\diam(G)$ denote the diameter of $G$, i.e.~$\diam(G) = \max_{x,y\in V(G)} d_G(x,y)$.
We refer to a path between two vertices $x,y\in V(G)$ as a $x$-$y$ path.

\paragraph{Cuts and $L_1$ embeddings.}
A cut of a graph $G$ is a partition of $V(G)$ into $(S,\bar{S})$---we sometimes
refer to a subset $S \subseteq V$ as a cut as well. A cut
gives rise to a pseudometric; using indicator functions, we can write the cut
pseudometric as $\rho_S(x,y) = |\Ind_S(x) - \Ind_S(y)|$. A central fact
%AJ: Two instances of pseudometric above changed from semi-metric
is that embeddings of finite metric spaces into $L_1$ are equivalent to sums
of positively weighted cut metrics over that set (for a simple proof
of this see \cite{GeoCuts}). %Thus, we only need to specify a set of
%cuts in order get an embedding into $L_1$.

A {\em cut measure on $G$} is a function $\mu : 2^V \to \mathbb R_+$
for which $\mu(S) = \mu(\bar S)$ for every $S \subseteq V$.
Every cut measure gives rise to an embedding $f : V\to L_1$ for which
\begin{equation}
\label{eq:cutmeasure}
\|f(u)-f(v)\|_1 = \int |\1_S(u)-\1_S(v)|\,d\mu(S),
\end{equation}
where the integral is over all cuts $(S, \bar S)$.
Conversely, to every embedding $f : V \to L_1$, we can
associate a cut measure $\mu$ such that \eqref{eq:cutmeasure} holds.
%We will use this correspondence freely in what follows.

\paragraph{Non-positively curved spaces.}
%\paragraph{Non-positive curvature in the sense of Busemann}

We will describe our proof using the definition of non-positive curvature in the sense of Busemann.
We give here a brief overview of some of the relevant terminology, and we refer the reader to \cite{papadopoulos2005metric,thurston1997three} for a more detailed exposition.
A metric space $(X,d)$ is called \emph{geodesic} if for every pair of points there exists a geodesic joining them.
%Unless stated otherwise we will consider geodesics that are minimizing, i.e.~such that $I=[0,d(x,y)]$, and for any $z,z'\in I$, we have $d(f(z),f(z')) = |z-z'|$.
We say that $(X,d)$ is non-positively curved, if for any pair of affinely parameterized geodesics $\gamma:[a,b]\to X$, $\gamma':[a',b']\to X$, the map $D_{\gamma,\gamma'}:[a,b] \times [a',b'] \to \mathbb{R}$ defined by
\[
D_{\gamma,\gamma'}(t,t') = d(\gamma(t), \gamma'(t))
\]
is convex.
As we show, this property is sufficient to obtain constant-distortion embeddings of simply-connected surfaces into $L_1$.

%----

\paragraph{Lipschitz partitions.}

Let $(X,d)$ be a metric space.
A distribution ${\cal F}$ over partitions of $X$ is called \emph{$(\beta,\Delta)$-Lipschitz} if every partition in the support of ${\cal F}$ has only clusters of diameter at most $\Delta$, and for every $x,y\in X$, 
\[
\Pr_{C\in {\cal F}}[C(x)\neq C(y)] \leq \beta \cdot \frac{d(x,y)}{\Delta}.
\]
We denote by $\beta_{(X,d)}$ the infimum $\beta$ such that for any $\Delta>0$, the metric $(X,d)$ admits a $(\Delta, \beta)$-Lipschitz random partition, and we refer to $\beta_{(X,d)}$ as the \emph{modulus of decomposability} of $(X,d)$.
The following theorem is due to Klein, Plotkin, and Rao \cite{DBLP:conf/stoc/KleinPR93}, and Rao \cite{DBLP:conf/compgeom/Rao99}.

\begin{theorem}[\cite{DBLP:conf/stoc/KleinPR93}, \cite{DBLP:conf/compgeom/Rao99}]\label{thm:KPR}
For any planar graph $G$, we have $\beta_{(V(G), d_G)} = O(1)$.
\end{theorem}

\paragraph{Stochastic embeddings.}
A mapping $f : X \to Y$ between two metric spaces $(X,d)$ and $(Y,d')$
is {\em non-contracting} if $d'(f(x),f(y)) \geq d(x,y)$ for all $x,y \in X$.
If $(X,d)$ is any finite metric space, and $\mathcal Y$
is a family of finite metric spaces, we say that {\em $(X,d)$ admits a stochastic $D$-embedding into $\mathcal Y$} if there exists a random metric space $(Y,d') \in \mathcal Y$ and a random
non-contracting mapping $f : X \to Y$ such that for every $x,y \in X$,
\begin{equation}
\label{eq:expansion}
\mathbb E\left[\vphantom{\bigoplus} d'(f(x),f(y))\right] \leq D \cdot d(x,y).
\end{equation}
The infimal $D$ such that \eqref{eq:expansion} holds is the {\em distortion of
the stochastic embedding.}
For a graph $G$ and a graph family ${\cal F}$ we write $G\overset{D}{\leadsto} {\cal F}$ to denote the fact that $G$ stochastically embeds into a distribution over graphs in ${\cal F}$, with distortion $D$.
We also use the notation $G \leadsto {\cal F}$ to denote the fact that $G\overset{D}{\leadsto} {\cal F}$, for some universal constant $D\geq 1$.
We will use the following fact.
\begin{lemma}\label{lem:stochastic_L_1}
Let ${\cal F}$ be a family of graphs, such that every $H\in {\cal F}$ admits an embedding into $L_1$ with  distortion at most $\alpha\geq 1$.
Let $G$ be a graph, such that $G \overset{\beta}{\leadsto} {\cal F}$, for some $\beta\geq 1$.
Then, $G$ admits an embedding into $L_1$ with distortion at most $\alpha \beta$.
\end{lemma}

Let $G$ be a graph, and let $A\subseteq V(G)$.
The \emph{dilation} of $A$ is defined to be
\[
\dil_G(A) = \max_{u,v\in V(G)} \frac{d_{G[A]}(u,v)}{d_G(u,v)}
\]
For two graphs $G,G'$, a \emph{1-sum} of $G$ with $G'$ is a graph obtained by taking two disjoint copies of $G$ and $G'$, and identifying a vertex $v\in V(G)$ with a vertex $v'\in V(G')$.
For a graph family ${\cal X}$, we denote by $\oplus_1{\cal X}$ the closure of ${\cal X}$ under 1-sums.
\begin{lemma}[Peeling lemma \cite{DBLP:conf/stoc/LeeS09}]\label{lem:peeling}
Let $G$ be a graph, and $A\subseteq V(G)$.
Let $G'=(V(G),E')$ be a graph with $E'=E(G)\setminus E(G[A])$,
and let $\beta = \beta_{(V, d_{G'})}$ be the corresponding modulus of decomposability.
%Then, $G$ can be stochastically embedded with distortion $O(\beta \cdot \dil_G(A))$, into a distribution over graphs 
Then, there exists a graph family ${\cal F}$ such that $G\overset{D}{\leadsto} {\cal F}$, where $D=O(\beta \cdot \dil_G(A))$, and every graph in ${\cal F}$ is a 1-sum of isometric copies of the graphs $G[A]$ and $\left\{G[V\setminus A\cup \{a\}]\right\}_{a\in A}$.
%Then, $G\overset{D}{\leadsto} H$, where $D=O(\beta \cdot \dil_G(A))$, and $H$ is a 1-sum of isometric copies of the graphs $G[A]$ and $\left\{G[V\setminus A\cup \{a\}]\right\}_{a\in A}$.
%Furthermore, the embedding always has distortion at most $\dil_G(A)$ for pairs $x,y\in A$.
\end{lemma}

%\paragraph{Reduction to the negatively curved case}

\subsection{Organization}

The rest of the paper is organized as follows.
In Section \ref{sec:funnels} we show how to embed an arbitrary non-positively curved planar metric into an unweighted graph of special structure, called a \emph{funnel}.
In Section \ref{sec:pyramids} we show how to stochastically embed a funnel into a distribution over simpler graphs, called \emph{pyramids}.
In Section \ref{sec:monotone} we introduce some of the machinery that we will use when defining our embedding into $L_1$.
More specifically, we describe the basic operation of cuts that will allow to gradually modify a cut when computing our embedding.
Using this machinery, we describe our embedding in Section \ref{sec:embedding}.
Finally, in Section \ref{sec:distortion} we prove that the constructed embedding has constant distortion.

\section{A canonical representation of non-positively curved planar metrics}\label{sec:funnels}

In this section we show that non-positively curved planar metrics can be embedded with constant-distortion into a certain type of unweighted planar graphs that we call \emph{funnels}.
Intuitively, a funnel is obtained by taking the union of a tree having all its leaves at the same level, with a collection of cycles, where every cycle spans all the vertices in a single layer of the tree.

\begin{definition}[Funnel]
Let $G$ be an unweighted planar graph, and let $v\in V(G)$.
We say that $G$ is a \emph{funnel} with \emph{basepoint} $v$ if the following conditions are satisfied:
\begin{description}
\item{(1)}
There exists a collection of pairwise vertex-disjoint cycles $C_1,\ldots,C_{\Delta} \subset G$, such that $V(G) = \bigcup_{i=1}^{\Delta} V(C_i)$.
For notational convenience, we allow a cycle $C_i$ to consist of a single vertex, in which case it has no edges.
Moreover, we have $V(C_1) = \{v\}$.
We refer to each $C_i$ as a \emph{layer} of $G$.

\item{(2)}
For every $i\in \{2,\ldots,\Delta-1\}$, the graph $G\setminus V(C_i)$ has exactly two connected components, one with vertex set $\bigcup_{j=1}^{i-1} V(C_j)$, and another with vertex set $\bigcup_{j=i+1}^{\Delta} V(C_j)$.

\item{(3)}
For every $i\in \{2,\ldots,\Delta\}$, every $u\in V(C_i)$ has exactly one neighbor $u' \in V(C_{i-1})$.
We refer to $u'$ as the \emph{parent} of $u$.
In particular, $v$ is the parent of all vertices in $V(C_2)$.

\item{(4)}
For every $i\in \{1,\ldots,\Delta-1\}$, every $w\in V(C_i)$ has at least one neighbor $w' \in V(C_{i+1})$.
We refer to every such $w'$ as a \emph{child} of $w$.
\end{description}
Let $R$ be a path in $G$ between $v$, and a vertex $u\in V(C_{\Delta})$.
We say that $R$ is a \emph{ray}.
We denote by $\funnels$ the family of all funnel graphs.
Figure \ref{fig:funnel} depicts an example of a funnel.
\end{definition}

\begin{figure}
\begin{center}
\scalebox{0.8}{\includegraphics{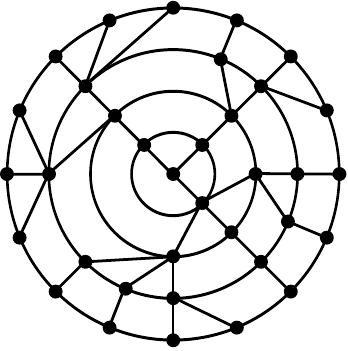}}
\caption{A funnel.\label{fig:funnel}}
\end{center}
\end{figure}

We will use the following two facts about metric spaces of non-positive curvature (see e.g.~\cite{papadopoulos2005metric}).

\begin{lemma}\label{lem:convexity}
Let $({\cal S},d)$ be a geodesic metric space of non-positive curvature.
Let $x^*,x,y\in {\cal S}$, and let $\gamma:[0,d(x,y)]\to {\cal S}$ be a geodesic between $x$, and $y$.
Then, the function $f:[0,1] \to \mathbb{R}$, with $f(t) = d(x^*,\gamma(t))$ is convex.
\end{lemma}

\begin{lemma}\label{lem:divergence}
Let $({\cal S},d)$ be a geodesic metric space of non-positive curvature, and let $x^*,x,y\in {\cal S}$.
Let $\gamma_x:[0,d(x^*,x)]\to {\cal S}$ be a geodesic between $x^*$, and $x$, and let $\gamma_y:[0,d(x^*,y)]\to {\cal S}$ be a geodesic between $x^*$, and $y$.
Then, the function $f:[0,1] \to \mathbb{R}$, with $f(t) = d(\gamma_x(t),\gamma_y(t))$ is non-decreasing.
\end{lemma}

Recall that for a metric space $(X,d)$, and some $r>0$, an \emph{$r$-net} in $(X,d)$ is a maximal subset $X'\subseteq X$ such that for any $x,y\in X'$, we have $d(x,y)\geq r$.

\begin{lemma}[Funnel representation]\label{lem:funnels}
Let ${\cal S}$ be a simply-connected surface, and let $d$ be a non-positively curved metric on ${\cal S}$.
Let $X \subset {\cal S}$ be a finite set of points.
Then, $(X,d)$ admits an embedding into a  funnel with constant distortion.
\end{lemma}
\begin{proof}
By scaling $d$, we may assume w.l.o.g.~that the minimum distance in $(X,d)$ is at least 1 (note that scaling $d$ results into a metric which is still of non-positive curvature).
Let $x^* \in {\cal S}$ be an arbitrary point.
For any $x\in {\cal S}$, let $\gamma(x)$ denote the unique geodesic between $x$, and $x^*$.
Let $r=1/8$.
For any integer $i\geq 0$, let
\[
D_i = \{x \in {\cal S} : d(x^*,x) \leq ir\}.
\]
Since $d$ is non-positively curved, we have that for every $i$, the set $D_i$ is a disk (see e.g.~\cite{papadopoulos2005metric}).
Let $\Gamma_i$ be the cycle in ${\cal S}$ bounding $D_i$.
Let $\Delta = \min\{i\in \mathbb{N} : X \subset D_i\}$.
Let $N_{\Delta}$ be an $r$-net in $\Gamma_{\Delta}$.
%For every $x\in N_{\Delta}$, let $\gamma(x)$ be a geodesic between $x$, and $x^*$.
Note that since $({\cal S},d)$ is non-positively curved, there exists a unique geodesic between any pair of points.
This implies that the subspace
\[
T = \bigcup_{x\in N_{\delta}} \gamma(x)
\]
is a (simplicial) tree.
For every $i\in \{0,\ldots,\Delta-1\}$, we define an $r$-net $N_i$ of $\Gamma_i$ as follows.
Suppose that $N_{i+1}$ is already defined.
Let $Y_i'$ be the set of points $p\in N_{\Delta}$ such that $\gamma(p)$ intersects $N_{i+1}$.
Let $N_i' = \Gamma_i \cap (\bigcup_{p\in Y_i'} \gamma(p))$.
Note that for any $x\in \Gamma_i$, there exists $y\in N_i'$ such that $d(x,y) <r$.
Therefore, we can set to be a maximal subset $N_i\subseteq N_i'$, such that $N_i$ is an $r$-net.
This concludes the definition of the sequence of subsets $N_0,\ldots,N_{\Delta}$. Note that $N_0 = \{x^*\}$.

We define a graph $G$, with $V(G) = \bigcup_{i=0}^{\Delta} N_i$.
The set of edges $E(G)$ is defined as follows.
For every $i\in \{1,\ldots,\Delta\}$, we add a unit-length edge $\{x,y\}\in E(G)$ for any two points $x,y\in N_i$, such that $x$, and $y$ appear consecutively in a clockwise traversal of $\Gamma_i$.
Moreover, for every $z\in N_i$, let $z'\in N_{\Delta}$ be such that $z\in \gamma(z')$.
Let $z''$ be the point in the intersection of $\gamma(z')$ with $\Gamma_{i-1}$.
If $z''\in N_{i-1}$, then we add the unit-length edge $\{z,z''\}$.
Otherwise, let $w$ be the first point in $N_{i-1}$ that we visit in a clockwise traversal of $\Gamma_{i-1}$ starting from $z''$.
We add the unit-length edge edge $\{z,w\}$.
This concludes the definition of the graph $G$.
It is straightforward to check that $G$ is a funnel with basepoint $x^*$.

We can now define an embedding $f:X \to V(G)$, by mapping every points $x\in X$ to its nearest neighbor in $V(G)$.
It remains to verify that $f$ has constant distortion.
Observe that the set $V(G)$ contains a $2r$-net in $D_{\Delta}$, and therefore for any $x\in X$, we have $d(x,f(x))<2r$.
Since the minimum distance in $X$ is at least 1, this implies that $f$ is an injection, and for any $x,y\in X$, we have $d(x,y) = \Theta(d(f(x),f(y)))$.
It therefore suffices to show that for any $x,y\in V(G)$, we have $d_G(x,y) = \Theta(d(x,y))$.

We first show that for any $x,y\in V(G)$, we have $d_G(x,y) = \Omega(d(x,y))$.
To that end, it suffices to show that for any edge $\{x,y\} \in E(G)$, we have $d(x,y) = O(d_G(x,y)) = O(1)$.

We consider first case where there exists $i\in \{1,\ldots,\Delta\}$ such that $x,y\in N_i$, and $x$, $y$ are consecutive in $\Gamma_i$.
Let $\alpha$ be the arc of $\Gamma_i$ between $x$, and $y$, that does not contain any other points in $N_i$.
By the triangle inequality, there exists $z \in \alpha$, such that $d(x,z)\geq d(x,y)\geq 2$, and $d(y,z) \geq d(x,y)\geq 2$.
Since $N_i$ is an $r$-net in $\Gamma_i$, it follows that there exists $z'\in N_i$, such that $d(z,z') < r$.
Let $\beta$ be the geodesic between $z$, and $z'$.
The arc $\beta$ intersects either $\gamma(x)$, or $\gamma(y)$.
Assume w.l.o.g.~that it intersects $\gamma(x)$ at some points $z''$.
By lemma \ref{lem:convexity} we have that 
as we travel along $\beta$, the distance to $x^*$ is a convex function.
This implies that $d(x,z'') \leq d(z,z')$.
We conclude that $d(x,y) \leq 2 d(x,z) \leq 2(d(x,z'') + d(z'',z)) \leq 2(d(x,z'') + d(z',z)) \leq 4 d(z,z') \leq 4r = O(1)$.

Next, we consider the case where $x\in N_i$, and $y\in N_{i+1}$, for some $i\in \{0,\ldots,\Delta\}$.
Let $y'$ be the point where $\gamma(y)$ intersects $\Gamma_i$.
Arguing as above, we have that $d(y',x) = O(1)$.
Therefore, $d(x,y) \leq d(y,y') + d(y',x) \leq r + O(1) = O(1)$.
This concludes that proof that for any edge $\{x,y\}\in E(G)$, we have $d(x,y) = O(1)$, and therefore for any $x,y\in V(G)$, we have $d_G(x,y) = \Omega(d(x,y))$.

It remains to show that for any $x,y\in V(G)$, we have $d_G(x,y) = O(d(x,y))$.
We consider first the case where there exists $i\in \{1,\ldots,\Delta\}$, such that $x,y\in N_i$ (the case $i=0$ is trivial since $N_0$ contains only $x^*$).
Let $\beta$ be a geodesic between $x$, and $y$.
By lemma \ref{lem:convexity}, we have $\beta\subset D_i$.
Let $x'$ be the unique point in $\gamma(x) \cap \Gamma_{i-1}$, and let $y'$ be the unique point in $\gamma(y) \cap \Gamma_{i-1}$.
By lemma \ref{lem:divergence} we have $d(x',y') \leq d(x,y)$.
Let $x''$ be the parent of $x$, and let $y''$ be the parent of $y$ in $G$.
Let $x'=z_1,\ldots,z_k=y'$ be the points in $N_{i-1}$ that appear between $x'$, and $y'$ along $\Gamma_{i-1}$.
For any $i\in \{1,\ldots,k\}$, pick a child $w_i$ of $z_i$, with $w_1=x$, and $w_k=y$.
For any $i\in \{1,\ldots,k\}$, the curve $\beta$ intersects $\gamma(w_i)$.
By the above discussion we have that the distance between any two such consecutive intersection points is $\Omega(1)$.
Therefore, $d(x,y) = \len(\beta) = \Omega(k)$.
The $x$-$y$ path in $G$ that visits the vertices $x z_1 \ldots z_k y$ in this order has length $k+2$, and therefore $d_G(x,y) = O(d(x,y))$.

Next, we consider the case where there exists $i\in \{1,\ldots,\Delta\}$, such that $x\in N_i$, and $y\in N_{i-1}$.
This case is identical to the case above, by replacing $y$ with $y''$.
We therefore also obtain $d_G(x,y) = O(d(x,y))$ in this case.

Finally, we consider the case of arbitrary points $x,y\in V(G)$.
Let $\beta$ be the geodesic between $x$, and $y$.
The curve $\beta$ can be decomposed into consecutive segments $\beta_1,\ldots,\beta_k$, such that every such segment is contained in (the closure of) $D_i \setminus D_{i-1}$, for some $i\in \{1,\ldots,\Delta \}$.
Consider such a segment $\beta_i$.
There exists $j,\ell\in \{0,\ldots,\Delta\}$, with $|j-\ell|\leq 1$, and such that $x_i\in \Gamma_j$, and $y_i \in \Gamma_{\ell}$.
Let $x_i'$ be the nearest neighbor of $x_i$ in $N_{j}$, and let $y_i'$ be the nearest neighbor of $y_i$ in $N_{\ell}$.
Since $N_j$ is a $O(1)$-net for $\Gamma_j$, and $N_{\ell}$ is a $O(1)$-net for $\Gamma_{\ell}$, we have $d(x_i',y_i') \leq d(x_i,y_i) + O(1) = O(d(x_i,y_i))$.
By the above analysis we have $d_G(x_i', y_i') = O(d(x_i',y_i'))$.
Therefore, we obtain $d_G(x_i,y_i) = O(d(x_i,y_i))$.
We conclude that $d_G(x,y) \leq \sum_i d_G(x_i,y_i) = O(\sum_i d(x_i,y_i)) = O(d(x,y))$, as required.
\end{proof}

\section{Cutting along a ray}\label{sec:pyramids}

We now show that every funnel admits a constant-distortion stochastic embedding into a distribution over simpler graphs, that we call \emph{pyramids}.
Intuitively, a pyramid is obtained by ``cutting'' a funnel along a ray.
The structure of pyramids will simplify the exposition of the embedding into $L_1$ that we describe in the subsequent sections.

\begin{definition}[Pyramid]
Let $G$ be an unweighted planar graph, let $v\in V(G)$, and let $\Delta\geq 1$ be an integer.
We say that $G$ is a \emph{pyramid} with \emph{basepoint} $v$, and of \emph{depth} $\Delta$ if the following conditions are satisfied:
\begin{description}
\item{(1)}
There exists a collection of pairwise vertex-disjoint paths $P_1,\ldots,P_{\Delta} \subset G$,
with $P_i = u_{i,1} \ldots u_{n_i,i}$,
such that $V(G) = \bigcup_{i=1}^{\Delta} V(P_i)$.
For notational convenience, we allow a path $P_i$ to consist of a single vertex, in which case it has no edges.
Moreover, we have $V(P_1) = \{v\}$.
We refer to each $P_i$ as a \emph{layer}.

\item{(2)}
For every $i\in \{2,\ldots,\Delta-1\}$, the graph $G\setminus V(P_i)$ has exactly two connected components, one with vertex set $\bigcup_{j=1}^{i-1} V(P_j)$, and another with vertex set $\bigcup_{j=i+1}^{\Delta} V(P_j)$.

\item{(3)}
For every $i\in \{2,\ldots,\Delta\}$, every $u\in V(P_i)$ has exactly one neighbor $u$ in $V(P_{i-1})$.
We refer to this neighbor as the \emph{parent} of $u$.
In particular, $v$ is the parent of all vertices in $V(P_2)$.

\item{(4)}
For every $i\in \{1,\ldots,\Delta-1\}$, every $w\in V(P_i)$ at least one neighbor $w'$ in $V(P_{i-1})$.
We refer to every such $w'$ as a \emph{child} of $w$.

\item{(5)}
For any $i\in \{1,\ldots,\Delta-1\}$, and for any $\{u_{i,j},u_{i+1,j'}\}, \{u_{i,t},u_{i+1,t'}\}\in E(G)$, we have $j\leq t \iff j'\leq t'$.
In other words, the ordering of the vertices in $P_{i+1}$ agrees with the ordering of their parents in $P_i$.
%We remark that although this condition is vacuous in the case of funnels due to planarity, this is not the case for pyramids.
\end{description}
We say that a path $R$ in $G$ between $v$, and a vertex $u\in V(P_\Delta)$, is a \emph{ray}.
We denote by $\pyramids$ the family of all pyramid graphs.
Figure \ref{fig:pyramid} depicts an example of a pyramid.
\end{definition}

\begin{definition}[Skeleton of a pyramid]
Let $G$ be a pyramid with basepoint $v\in V(G)$.
We define the \emph{skeleton} of $G$ to be a tree $T$, with $V(T) = V(G)$, with root $v$, and with 
\[
E(T) = \left\{\{x,y\} \in {V(G) \choose 2} : x \text{ is the parent of } y \right\}.
\]
For any $x,y\in V(G)$, we denote by $\nca$ the nearest common ancestor of $x$, and $y$ in $T$.
We also define for any $x\in V(G)$,
\[
\depth(x) = d_T(v,x)+1.
\]
Figure \ref{fig:pyramid} depicts an example of a skeleton.
\end{definition}

\begin{figure}
\begin{center}
\scalebox{0.8}{\includegraphics{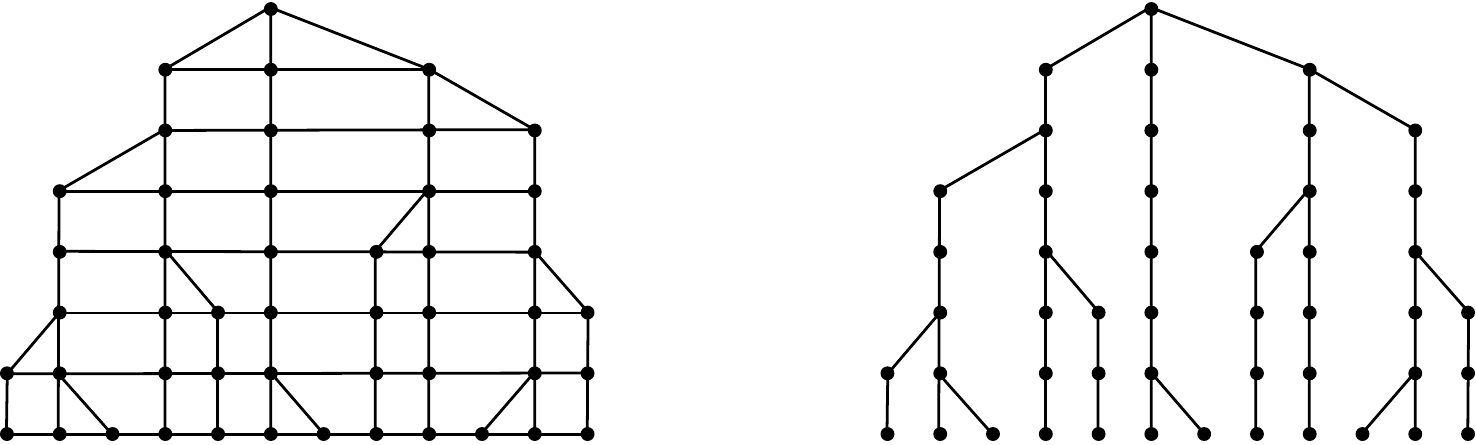}}
\caption{A pyramid (left), and its skeleton (right).\label{fig:pyramid}}
\end{center}
\end{figure}

\begin{definition}[$\prec$]
For any $i\in \{1,\ldots,\Delta\}$, for any $u_{i,j}, u_{i,j'} \in V(P_i)$, with $j<j'$, we write $u_{i,j} \prec u_{i,j'}$.
Moreover, for any $x,y\in V(G)$, such that $x$, and $y$ do not lie on the same ray, let $z=\nca(x,y)$, and let $x'$ (resp.~$y'$) be the child of $z$ in the $z$-$x$ (resp.~$z$-$y$) path in $T$.
Then, we write $x \prec y$ if and only if $x' \prec y'$.
Finally, for any $x'',y''\in V(G)$, we write $x'' \preceq y''$ if and only if either $x'' \preceq y''$, or $x''$, and $y''$ lie on the same ray.
\end{definition}

\begin{lemma}[Pyramid representation]\label{lem:pyramids}
For every funnel $G$, we have $G \leadsto \oplus_1\{\pyramids\}$.
\end{lemma}
\begin{proof}
Let $G$ be a funnel with basepoint $x^* \in V(G)$, and depth $\Delta$.
Let $R$ be a ray in $G$.
Replace $R\setminus x^*$ by a $\Delta \times 4$ grid $H$.
Clearly, this results into an embedding of $G$ into a funnel $G'$ with distortion $O(1)$.
Let $R'$ be the union of the two central columns of $H$, and let $A = R' \cup \{x^*\}$.
Observe that $\dil_{G'}(A) = 1$.
Applying lemma \ref{lem:peeling} on $G'$ and the set $A$, we obtain a stochastic embedding of $G'$ into a distribution of graphs ${\cal D}$.
Since $G'$ is planar, and $\dil_{G'}(A) = 1$, it follows by Theorem \ref{thm:KPR} that the distortion of the resulting stochastic embedding is $O(1)$.
Every graph in the support of ${\cal D}$ is obtained via 1-sums of $G'[A]$, with $G'[V\setminus A \cup \{a\}]$, for some $a\in A$.
The graph $G'[A]$ is a $\Delta\times 2$ grid, with the basepoint $x^*$ connected to the two vertices in the top row, and is therefore a pyramid.
For any $a\in A$, the graph $G'[V\setminus A \cup \{a\}]$ is obtained from $G'$ by cutting along a ray, and is therefore also a pyramid.
This concludes the proof.
\end{proof}

\section{Monotone cuts}\label{sec:monotone}

In this section we describe the family of cuts, that we will use when defining our embedding into $L_1$.
These are cuts that we call \emph{monotone}, and intuitively correspond to sets that only cross every ray at most once.
We also describe a specific ``shifting'' operation that will allow us to modify a cut in order to adapt to the finer geometry of a given space.

\begin{definition}[Monotone cut]
Let $G$ be a pyramid with basepoint $v\in V(G)$, and
let $S\subseteq V(G)$.
We say that $S$ is \emph{$v$-monotone} (or monotone when $v$ is clear from the context) if $v\in S$, and 
for any ray $R$ in $G$, $R\cap S$ is a prefix of $R$.
In particular, this implies that $G[S]$ is a connected subgraph
%We also say that the cut $(S, V(G) \setminus S)$ is monotone.
(see Figure \ref{fig:monotone}).
\end{definition}

\begin{definition}[Boundary of a monotone cut]
Let $S\subseteq V(G)$ be a monotone cut.
We define the \emph{vertex boundary} of $S$, denoted by $\partial_V S$, to be the set of all $u\in S$, such that all children of $u$ are not in $S$.
We also define the \emph{edge boundary} of $S$, denoted by $\partial_E S$, to be
\[
\partial_E S = \{\{x,y\}\in E(G) : x,y\in \partial_V S \text{ and } \depth(x) = \depth(y)\}.
\]
Finally, we define the graph $\partial S = (\partial_V S, \partial_E S)$ (see Figure \ref{fig:monotone}).
\end{definition}

\begin{figure}
\begin{center}
\scalebox{0.8}{\includegraphics{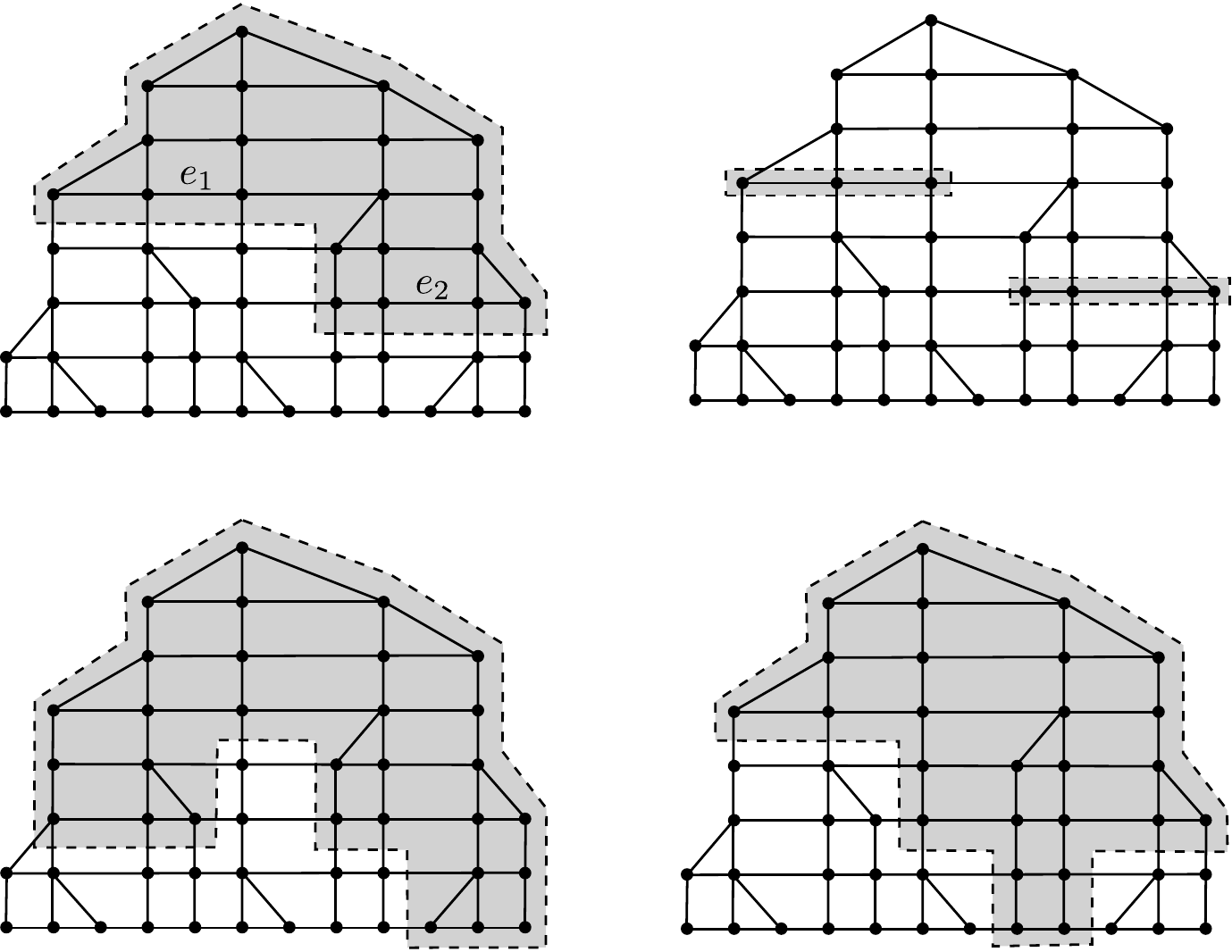}}
\caption{A monotone cut (top-left), its boundary (top-right), its odd $(2,\{e_1,e_2\})$-shift (bottom-left), and its even $(2,\{e_1,e_2\})$-shift (bottom-right).\label{fig:monotone}}
\end{center}
\end{figure}

\iffalse
\begin{definition}[Segment decomposition]
Let $S\subseteq V(G)$ be a monotone cut.
Then, $\partial_V S$ can be uniquely decomposed into a collection of pairwise vertex-disjoint horizontal paths.
We refer to this decomposition ${\cal Q}$ as the \emph{segment decomposition} of $\partial_V S$, and to every path $Q\in {\cal Q}$ as a \emph{segment} of the decomposition.
\end{definition}
\fi

\begin{definition}
Let $G$ be a pyramid, let $T$ be the skeleton of $G$.
Let $u\in V(G)$, and $r\geq 0$.
Then, we denote by $\widetilde{N}(u, r)$ the set of all vertices $w\in V(G)$, such that $u$ is an ancestor of $w$ in $T$, and $d_T(u,w) \leq r$.
\end{definition}

\begin{definition}[Odd/even shift of a monotone cut]
Let $S\subseteq V(G)$ be a monotone cut, let $r>0$, and $Z\subseteq \partial_E S$.
Let $Z = \{\{x_i,y_i\}\}_{i=1}^{k}$, with
\[
x_1 \prec y_1 \preceq x_2 \prec y_2 \preceq \ldots \preceq x_k \prec y_k.
\]
Let $\{V_i\}_{i=1}^{k+1}$ be a decomposition of $\partial_V S$, with
$V_1=\{u\in \partial_V S : u \preceq x_1\}$,
$V_{k+1}=\{u\in \partial_V S : y_k \preceq u\}$,
and
for any $i\in \{2,\ldots,k\}$, 
$V_{i}=\{u\in \partial_V S : y_i \preceq u \preceq x_{i+1}\}$.
We define a partition $\partial_V S = V_{\text{odd}} \cup V_{\text{even}}$, by setting
$V_{\text{odd}} = \bigcup_{i=1}^{\lceil t/2 \rceil} V(Q_{2i-1})$,
and
$V_{\text{even}} = \bigcup_{i=1}^{\lfloor t/2 \rfloor} V(Q_{2i})$.
We define the \emph{odd $(r,Z)$-shift} of $S$ to be the cut $S_{\text{odd}}$ given by 
\[
S_{\text{odd}} = S \cup \bigcup_{u \in V_{\text{odd}}} \widetilde{N}(u,r).
\]
Similarly, we define the \emph{even $(r,Z)$-shift} of $S$ to be the cut $S_{\text{even}}$ given by 
\[
S_{\text{even}} = S \cup \bigcup_{u \in V_{\text{even}}} \widetilde{N}(u,r).
\]
We  say that a cut $S'$ is a $(r,Z)$-shift of $S$, if it is either the odd, or the even $(r,Z)$-shift of $S$ (see Figure \ref{fig:monotone} for an example).
\end{definition}

\section{The embedding}\label{sec:embedding}
In this section we present a constant-distortion embedding of pyramids into $L_1$.
Combining with lemmas \ref{lem:stochastic_L_1}, \ref{lem:funnels}, \& \ref{lem:pyramids}, this implies that every planar metric of non-positive curvature embeds into $L_1$ with constant distortion.

Let $G$ be a pyramid, with basepoint $v\in V(G)$.
Let $\Delta \geq 1$ be the depth of $G$, and let $\delta = \lceil \log \Delta \rceil$.
It will be convenient for our exposition to isometrically embed $G$ into a larger pyramid $G'$, with depth $\Delta' = O(\Delta)$, as follows.
The pyramid $G'$ contains a copy of $G$, and a new basepoint $v'$, that is connected to the basepoint $v$ of $G$ via a path of length $2 \Delta$, resulting into a pyramid of depth $\Delta' = 3\Delta$.
\begin{center}
\scalebox{0.5}{\includegraphics{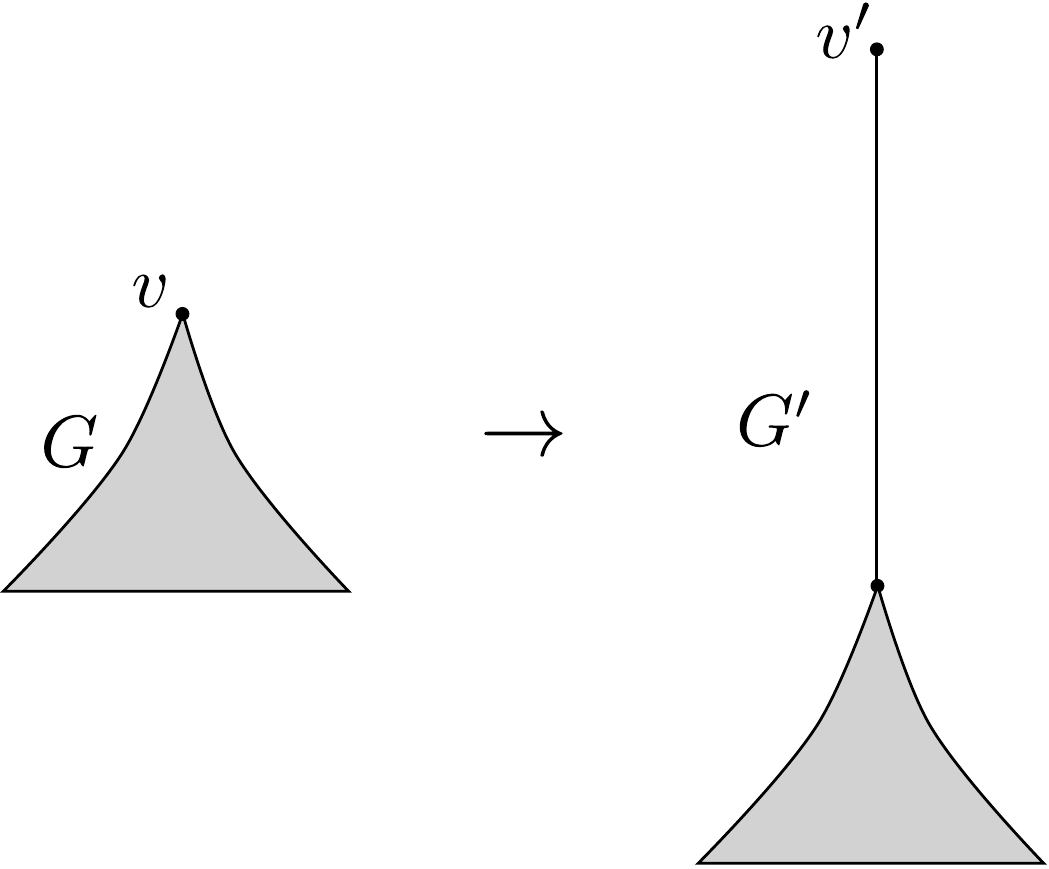}}
\end{center}
We will then compute an embedding for $G'$, and prove that its restriction on $G$ has small distortion.
We remark that our embedding will have unbounded distortion for points in $G'$ close to $v'$ (more precisely, pairs of vertices at distance $\eps$ from $v'$, will be distorted by a factor of $O(1/\eps)$).
However, this does not affect our result, since we only case about distances in $G$, which lies far from $v'$.

\begin{definition}[Evolution of a monotone cut]
%Let $G$ be a pyramid with basepoint $v\in V(G)$, and of depth $\Delta \geq 1$.
Let $r>0$, and let $S\subseteq V(G')$ be a monotone cut.
The \emph{$r$-evolution} of $S$ is a probability distribution ${\cal D}$ over monotone cuts, defined by the following random process.
Let
\[
Y = \{\{x,y\} \in \partial_E S : \depth(x) - \depth(\nca(x,y)) \in [r, 6r) \}.
\]
Pick a random subset $Y' \subseteq Y$, by choosing every $e\in Y$ independently, with probability $1/r$.
We probability $1/2$, let $S'$ be the odd $(r,Y')$-shift of $S$, and otherwise let $S'$ be the even $(r,Y')$-shift of $S$.
The resulting random cut $S'$ defines the distribution ${\cal D}$.
\end{definition}

Let ${\cal M}$ be the set of all $v$-monotone cuts in $G'$.
We inductively define a sequence $\{\mu_i\}_{i=0}^{\delta+1}$, where each ${\cal \mu}_i$ is a probability distribution over ${\cal M}$.
We define $\mu_{0}$ as follows.
Let $P_1,\ldots,P_{\Delta'}$ be the layers of $G'$.
For any $j\in \{1,\ldots,\Delta'\}$, let $X_i = \bigcup_{t=1}^j V(P_t) = \ball(v',i)$.
Let $\mu_{0}$ be the uniform distribution over the collection of cuts $X_1,\ldots,X_{\Delta'}$.

%For any $S\in {\cal M}$, we have
%\begin{align*}
%\Pr_{S' \in {\cal \mu}_0}[S'=S] = \left\{ \begin{array}{ll}
%1/\Delta & \mbox{ if there exists }  i\in \{1,\ldots,\Delta\} \mbox{ s.t.~} S = \bigcup_{j=1}^i V(P_j)\\
% 0 & \mbox{ otherwise}
%\end{array} \right.
%\end{align*}

For any $i\geq 0$, given $\mu_i$, we inductively define $\mu_{i+1}$ via the following random process:
We first pick a random cut $S_i$ according to ${\cal \mu}_i$.
Let ${\cal D} = {\cal D}(S_i)$ be the $\Delta / 3^i$-evolution of $S_i$.
We pick a random cut $S_{i+1}$ according to ${\cal D}$.
The resulting random variable $S_{i+1}$ defines the probability distribution $\mu_{i+1}$.

We define the embedding $f$ induced by the probability distribution $\mu_\delta$, and the embedding $f_0$ induced by the probability distribution $\mu_0$.
Finally, we set the resulting embedding to be
\[
g = f \oplus f_0,
\]
i.e.~the concatenation of the embeddings $f$, and $f_0$.
In the next section we show that the distortion of $g$ restricted on $G$ is bounded by some universal constant.

\section{Distortion analysis}\label{sec:distortion}

We now analyze the distortion of the embedding $g$ constructed in the previous section.

\subsection{Distortion of vertical pairs of points}

\begin{lemma}\label{lem:uniform_vertical}
Let $u\in V(G)$, with $\depth(u)<\Delta'$, and let $i\in \{1,\ldots,\delta\}$.
Then, $\Pr[u \in \partial_V S_i] = 1/\Delta'$.
\end{lemma}
\begin{proof}
The proof is by induction on $i$.
For $i=0$, the assertion holds since $\mu_0$ is the uniform distribution over the cuts $X_1,\ldots,X_{\Delta'}$.
Suppose next that $i>0$.
Let $r = \Delta / 3^{i-1}$, and let $u'$ be the ancestor of $u$ in $T$, with $d_T(u,u) = r$.
Fix some $S_{i-1}$ in the support of $\mu_{i-1}$, and suppose that $S_i$ is sampled from the $r$-evolution of $S_{i-1}$.
This means that we first sample a set of edges $Y$, and for any such $Y$ we set $S_i$ to be the odd $(r,Y)$-shift of $S_{i-1}$ with probability $1/2$, or otherwise we set $S_i$ to be the even $(r,Y)$-shift of $S_{i-1}$.
Therefore, we have have $u\in \partial_V S_i$, only if either $u\in \partial_V S_{i-1}$, or $u'\in \partial_V S_{i-1}$.
Conditioned on either of these two events, and for any $Y$, exactly one of the odd/even shifts of $S_{i-1}$ has $u$ in its boundary.
This implies that 
\begin{align*}
\Pr[u \in \partial_V S_{i}] &= \Pr[u \in \partial_V S_{i} | u\in \partial_V S_{i-1}] \cdot \Pr[u\in \partial_V S_{i-1}]\\
 &\phantom{=} + \Pr[u\in \partial_V S_{i} | u' \in \partial_V S_{i-1}] \cdot \Pr[u'\in \partial_V S_{i-1}]\\
 &= \frac{1}{\Delta'} \cdot \frac{1}{2} + \frac{1}{\Delta'} \cdot \frac{1}{2}\\
 &= 1/\Delta',
\end{align*}
as required.
\end{proof}

\begin{lemma}\label{lem:vertical_pairs}
Let $x,y\in V(G)$, such that $x,y$ lie on the same ray.
Then, $\|f(x)-f(y)\|_1 = d_G(x,y)/\Delta$.
\end{lemma}
\begin{proof}
Let $R$ be the ray containing both $x$, and $y$.
Let $R'$ be the subpath of $R$ between $x$, and $y$, including $x$, and excluding $y$.
By the monotonicity of ${\cal S}_\delta$, it follows that $\mathbf{1}_{S_{\delta}}(x) \neq \mathbf{1}_{S_{\delta}}(y)$, if and only if there exists $z\in V(R')$, such that $z\in \partial_V S_{\delta}$.
Since these events are disjoint for different $z$, we obtain by lemma \ref{lem:uniform_vertical} that $\|f(x)-f(y)\|_1 = \Pr[\mathbf{1}_{S_{\delta}}(x) \neq \mathbf{1}_{S_{\delta}}(y)] = |V(R')| / \Delta' = d_G(x,y) / \Delta'$, as required.
\end{proof}

\subsection{Distortion of horizontal pairs of points}
We now bound the distortion on pairs of vertices $x,y\in V(G)$ that lie on the same layer of $G'$, i.e.~such that $\depth(x)=\depth(y)=h$.
Let $d_G(x,y)=L$.
Let also $h'=\depth(\nca(x,y))$.
We assume w.l.o.g.~that $x \preceq y$.
Let $P$ be the subpath of $P_h$ between $x$, and $y$.

Let
\[
E_{\text{top}} = \left\{\{z,w\} \in E(P) : \depth(\nca(z,w)) \leq h - L/2 \right\},
\]
and
\[
E_{\text{bottom}} = E(P) \setminus E_{\text{top}}.
\]

\begin{lemma}\label{lem:E_top}
%Let $G$ be a pyramid, and let $x,y\in V(G)$, with $\depth(x) = \depth(y) = h$.
%Let $P$ be the subpath of $P_h$ between $x$, and $y$.
%Suppose that $d_G(x,y) = L$.
$|E_{\text{top}}| \leq L$.
\end{lemma}
\begin{proof}
Suppose, to the contrary, that $|E_{\text{top}}| > L$.
For any $i\in \{h-L/2,\ldots,h\}$, let $Z_i$ be the subpath of $P_i$ between the ancestor of $x$, and the ancestor of $y$ in $P_i$.
For any $e=\{z,w\}\in E_{\text{top}}$, with $z\prec w$, let $R_z$ be a ray containing $z$, and let $W_e$ be the subpath of $R_e$ contained between $P_{h-L/2}$, and $P_h$.
\begin{center}
\scalebox{0.5}{\includegraphics{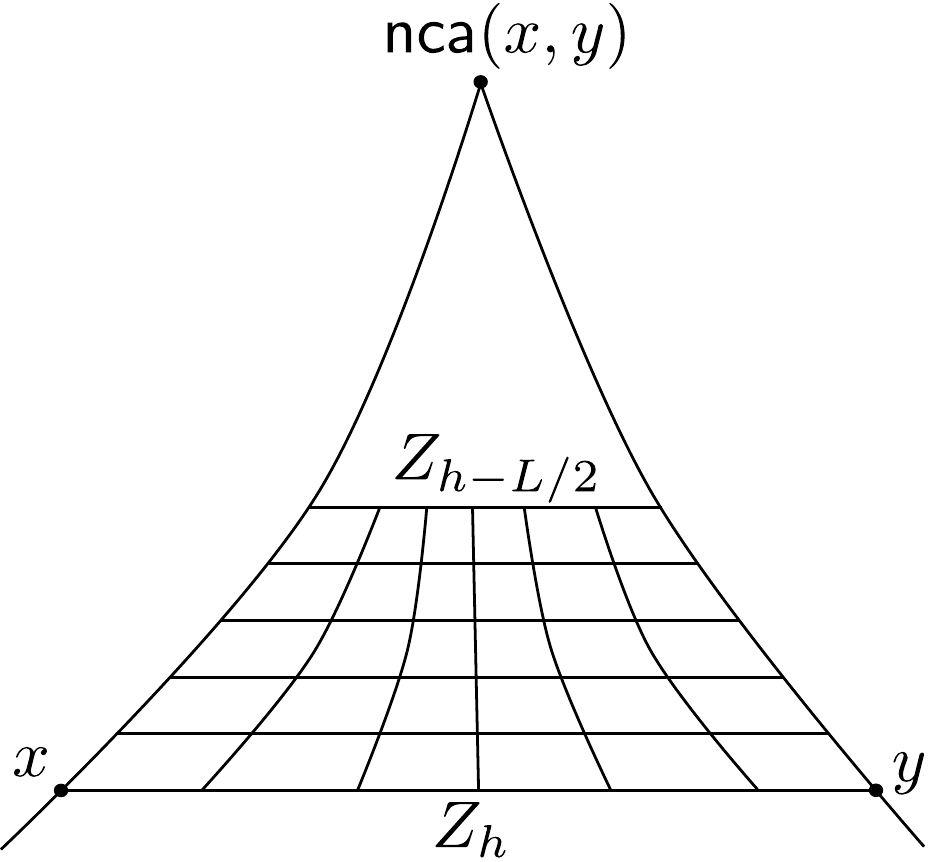}}
\end{center}
The union of all these paths $\left(\bigcup_{i}Z_i\right) \cup \left( \bigcup_e W_e \right)$ forms a $(L/2+1) \times L$ grid minor in $G'$, with $x$, and $y$ being the bottom-left, and bottom-right vertices respectively.
Since the $x$-$y$ shortest path in $G'$ is contained in $\ball(v',h)$, this implies that $d_G(x,y) > L$, which is a contradiction.
\end{proof}

\iffalse
\begin{definition}[Wedge, and top-wedge]
%Let $G$ be a pyramid, and let $x,y\in V(G)$, with $\depth(x) = \depth(y)$, and let $d_G(x,y)=L$.
%Assume w.l.o.g.~that $x\preceq y$.
Let
\[
X = \{u\in V(G) : \depth(\nca(x,y)) \leq \depth(u) \leq \depth(x) \text{ and } x \preceq u \preceq y\}.
\]
We define the \emph{$\{x,y\}$-wedge} to be the subgraph of $G$ induced by $X$.
Similarly, let
\[
Y = \{u\in V(G) : \depth(\nca(x,y)) \leq \depth(u) \leq \depth(x) - L/c_3 \text{ and } x \preceq u \preceq y\}.
\]
We define the \emph{$\{x,y\}$-top-wedge} to be the subgraph by $G$ induced in $Y$.
\end{definition}
\fi

Let $H$ be the subgraph of $G$ induced on the set of vertices
\[
V(H) = \{u\in V(G) : h' \leq \depth(u) \leq h \text{ and } x \preceq u \preceq y\}.
\]

\begin{definition}[Straight cut]
Let $i\in \{1,\ldots,\delta\}$, and $j\in \{1,\ldots,\Delta'\}$.
We say that $S_i$ is \emph{$j$-straight} if
%either $\partial S_i \cap H = \emptyset$, or
$\partial S_i \cap H \subseteq P_j$.
\end{definition}

Let $e=\{z,w\}\in E(P)$.
We say that an edge $e'=\{z',w'\}\in E(G)$ is an \emph{ancestor} of $e$, if $z'$ is an ancestor of $z$ in $T$, $w'$ is an ancestor of $w$ in $T$, and $\depth(z')=\depth(w')$.

\begin{definition}[Bend]
Let $e\in E(P)$.
We say that $e$ \emph{bends} $S_i$, if the following events happen.
\begin{description}
\item{(1)}
There exists $j\in \{1,\ldots,\Delta'\}$, such that
$S_i$ is $j$-straight.
\item{(2)}
Let $Y\subseteq \partial_E S_i$, such that $S_{i+1}$ is the $(r,Y)$-shift of $S_i$, for some $r>0$.
Then, there exists an ancestor of $e$ in $Y$.
\end{description}
\end{definition}

\begin{lemma}\label{lem:straight_upper}
Let $j \in \{h',\ldots,h\}$, and
let $i\in \{1,\ldots,\delta\}$.
Then, $\Pr[S_i \text{ is } j\text{-straight}] \leq 1/\Delta'$. 
\end{lemma}
\begin{proof}
Let $z$ be an arbitrary vertex in $V(P_j) \cap V(H)$.
Clearly, $S_i$ can only be $j$-straight if $z \in \partial_V S_i$.
Therefore, by lemma \ref{lem:uniform_vertical} we obtain
$\Pr[S_i \text{ is } j\text{-straight}] \leq \Pr[z \in \partial_V S_i] = 1/\Delta'$, as required.
\end{proof}

For any edge $e=\{z,w\}\in E(P)$, and for any $i\in \{1,\ldots,\delta\}$, let ${\cal E}(e,i)$ be the conjunction of the following two events:
\begin{itemize}
\item
${\cal E}_1(e,i)$:
There exists $j$, such that the following event, denoted by ${\cal E}_1(e,i,j)$, holds: 
Intuitively, the event ${\cal E}_1(e,i,j)$ describes a necessary condition such that a bend of $S_i$ can potentially lead to a cut $S_{\delta}$ that separates $x$, and $y$.
Formally, we have that $S_i$ is $j$-straight, with
\begin{align}
\Delta / 3^i \leq  j - \depth(\nca(z,w)) < 6 \Delta / 3^i, \label{eq:event_1}
\end{align}
and
\begin{align}
h - j &\leq 2 \Delta / 3^i. \label{eq:event_2}
\end{align}
\item
${\cal E}_2(e,i)$:
$e$ bends $S_i$.
\end{itemize}

\begin{lemma}\label{lem:if_cut_then_bent}
Suppose that $\mathbf{1}_{S_{\delta}}(x) \neq \mathbf{1}_{S_{\delta}}(y)$.
Then, there exists $e\in E(P)$, and $i\in \{1,\ldots,\delta\}$, such that the event ${\cal E}(e,i)$ occurs.
\end{lemma}
\begin{proof}
Recall that by the definition of $\mu_0$, the cut $S_0$ is $j_0$-straight, for some $j_0\in \{1,\ldots,\Delta'\}$.
Since $\mathbf{1}_{S_{\delta}}(x) \neq \mathbf{1}_{S_{\delta}}(y)$, it follows that for all $j_{\delta}\in \{1,\ldots,\Delta'\}$, the cut $S_{\delta}$ is not $j_{\delta}$-straight.
Let $i^*\in \{0,\ldots,\delta-1\}$ be the smallest integer such that
for all $j\in \{1,\ldots,\Delta'\}$, the cut $S_{i^*+1}$ is not $j$-straight.
This means that $S_{i^*}$ is $j^*$-straight, for some $j^*\in \{1,\ldots,\Delta'\}$.
Therefore, there exists $e=\{z,w\}\in E(P)$, such that $e$ bends $S_{i^*}$, which means that the event ${\cal E}_2(e,i^*)$ occurs.
It suffices to show that the event ${\cal E}_1(e,i^*,j^*)$ also occurs.
We have established that $S_{i^*}$ is $j^*$-straight, so its remains to show that \eqref{eq:event_1} \& \eqref{eq:event_2} hold.
Condition \eqref{eq:event_1} follows immediately form the fact that $e$ bends $S_{i^*}$, and $S_{i^*+1}$ is the $(Y,\Delta/3^{i^*})$-shift of $S_{i^*}$, with $e\in Y$.
Since $S_{i^*}$ is $j^*$-straight, we have $S_{i^*} \subseteq \ball(v',j^*)$.
The cut $S_{\delta}$ is obtained from $S_{i^*}$ via a sequence of $(Y,r)$-shifts, with exponentially decreasing values of $r$.
This implies 
$S_{\delta} \subseteq \ball(v',t)$, for some $t\leq j^* + \sum_{i=i^*}^\delta \Delta / 3^i < j^* + 2\Delta / 3^*$.
Since $\mathbf{1}_{S_{\delta}}(x) \neq \mathbf{1}_{S_{\delta}}(y)$, we have $t>h$, and therefore $h-j^* < 2\Delta / 3^{i^*}$, which implies \eqref{eq:event_2}, and concludes the proof.
\end{proof}

\begin{lemma}[Expansion of horizontal pairs]\label{lem:horizontal_expansion}
Let $x,y\in V(G)$, such that $\depth(x)=\depth(y)$.
Then,
$\|f(x)-f(y)\|_1 = O( d(x,y)/\Delta')$.
\end{lemma}
\begin{proof}
Let ${\cal E}_{\text{top}}$ denote the event that there exists $e\in E_{\text{top}}$, and $i\in \{1,\ldots,\delta\}$, such that ${\cal E}(e,i)$ occurs.
Similarly, let ${\cal E}_{\text{bottom}}$ denote the event that there exists $e\in E_{\text{bottom}}$, and $i\in \{1,\ldots,\delta\}$, such that ${\cal E}(e,i)$ occurs.
By lemma \ref{lem:if_cut_then_bent} we have
\begin{align*}
\|f(x) - f(y)\|_1 &= \Pr[\mathbf{1}_{S_{\delta}}(x) \neq \mathbf{1}_{S_{\delta}}(x)]\\
 &\leq \Pr[{\cal E}_{\text{top}}] + \Pr[{\cal E}_{\text{bottom}}].
\end{align*}
Let us bound the two latter quantities separately.

We first bound $\Pr[{\cal E}_{\text{top}}]$.
Let $e=\{z,w\}\in E_{\text{top}}$, and $i\in \{1,\ldots,\delta\}$.
Let $h' = \depth(\nca(\{z,w\}))$.
%We first argue that for sufficiently large $i$, the event ${\cal E}(e,i)$ cannot occur.
Recall that by the definition of ${\cal E}_1(e,i,j)$, in order for ${\cal E}_1(e,i,j)$ to occur for some $j$, we must have by \eqref{eq:event_1} that $j-h' \leq 6\Delta / 3^i$, and by \eqref{eq:event_2} that $h-j \leq 2\Delta / 3^i$.
We therefore obtain that
\[
h-h' = h-j + j-h' = O(\Delta / 3^i).
\]
%By the first inequality in \eqref{eq:event_1}, this implies that $j-h' = \Omega(h-h')$.
%Since $3^i \leq 4\Delta/(h-h')$, we have that conditioned on ${\cal E}_1(e,i,j)$, for some $j$, 
Note that $S_{i+1}$ is the $(Y,r)$-shift of $S_i$, for some $r = \Delta / 3^i$, and for random some $Y \subseteq \partial_E S_{i}$.
In order for the edge $e$ to bend $S_i$, its must be the case that its unique ancestor (if it exists) in $\partial_E S_i$ is chosen in $Y$.
Every edge in chosen in $Y$ with probability at most $1/r$.
Therefore, for any $j$, and for any $i$, we have
\[
\Pr[{\cal E}_2(e,i) | {\cal E}_1(e,i,j)] \leq 3^i / \Delta.
\]
Moreover, ${\cal E}_1(e,i,j)$ can occur only if $j\in \{h',\ldots,h\}$.
For each such value $j\in \{h',\ldots,h\}$, and for any $i$, we have by lemma \ref{lem:straight_upper} that 
\[
\Pr[{\cal E}_1(e,i,j)] = O(1/\Delta).
\]
To summarize, we have
\begin{align}
\Pr[{\cal E}_{\text{top}}] &\leq \sum_{e\in E_{\text{top}}} \sum_{i\in \{1,\ldots,\delta\}} \Pr[{\cal E}(e,i)] \notag \\
 &\leq \sum_{e\in E_{\text{top}}}  \sum_{j\in \{h',\ldots,h\}}  \sum_{i\in \{1,\ldots,\delta\}}  \Pr[{\cal E}_2(e,i) | {\cal E}_1(e,i,j)] \cdot \Pr[{\cal E}_1(e,i,j)] \notag \\
 &\leq \sum_{e\in E_{\text{top}}}  \sum_{j\in \{h',\ldots,h\}}  \sum_{i\in \{1,\ldots,\delta\}}  \frac{3^i}{\Delta} O(1/\Delta) \notag \\
 &\leq \sum_{e\in E_{\text{top}}}  \sum_{j\in \{h',\ldots,h\}}    O(1/(h-h')) \cdot  O(1/\Delta) \notag \\
 &\leq \sum_{e\in E_{\text{top}}}  O(1/\Delta) \notag \\
 &= O(|E_{\text{top}}|/\Delta) \notag \\
 &= O(L/\Delta'). \label{eq:distortion_upper1}
\end{align}

We next bound $\Pr[{\cal E}_{\text{bottom}}]$.
Let $\{1,\ldots,\delta\}$, $j\in \{1,\ldots,\Delta'\}$, $e\in E_{\text{bottom}}$, such that both ${\cal E}_1(e,i,j)$, and ${\cal E}_2(e,i)$ occur.
As above, let $e=\{z,w\}$, and $h'=\depth(\nca(z,w))$.
Then, we must have $h' \leq j \leq h$, which implies $j-h' \leq h-h' \leq L/2$.
Since $S_{i+1}$ is the $(r,Y)$-shift of $S_i$ for some $Y\subseteq E(P)$, with $r=\Delta/3^i$, we obtain that
$j-h' \in [r, 6 r)$, which implies
$3^i \geq 2 \Delta/L$.
Let $R_x$ be the ray containing $x$, and let $\chi$ be the unique vertex in the intersection of $R_x$ with $\partial S_i$.
Let also $\chi'$ be the unique vertex in the intersection of $R_x$ with $\partial S_{\delta}$.
For every $i'\geq i$, the intersection of $\partial S_{i'}$ with $R_x$ moves by at most $\Delta/3^{i'}$ along $R_x$, and therefore
 $d_T(\chi, \chi') < 2\Delta / 3^{i'} = O(L)$.
Since $\chi \in P_j$, and $j\in [h',h]$, it follows that $\depth(\chi)$ can take at most $h'-h+1$ different values.
Therefore, $\chi'$ can only lie inside a subpath $R_x' \subseteq R_x$ of length $O(h'-h)$.
Applying lemma \ref{lem:uniform_vertical}, we obtain
\begin{align}
\Pr[{\cal E}_{\text{bottom}}] &\leq \Pr[\chi' \in R_x'] \notag \\
 &\leq |V(R_x')| / \Delta' \notag \\
 &= O(h'-h) / \Delta' \notag \\
 &= O(L/\Delta'). \label{eq:distortion_upper2}
\end{align}
Combining \eqref{eq:distortion_upper1} \& \eqref{eq:distortion_upper2} we conclude that $\|f(x)-f(y)\|_1 = O(L/\Delta') = O(d(x,y)/\Delta')$, as required.
\end{proof}

We now bound the contraction of $f$.
For any $i\in \{1,\ldots,\delta\}$, let
\[
J_i = \left\{h-\frac{\Delta}{3^i}, \ldots, h - \frac{2}{3}\cdot \frac{\Delta}{3^i}\right\}.
\]

\begin{center}
\scalebox{0.55}{\includegraphics{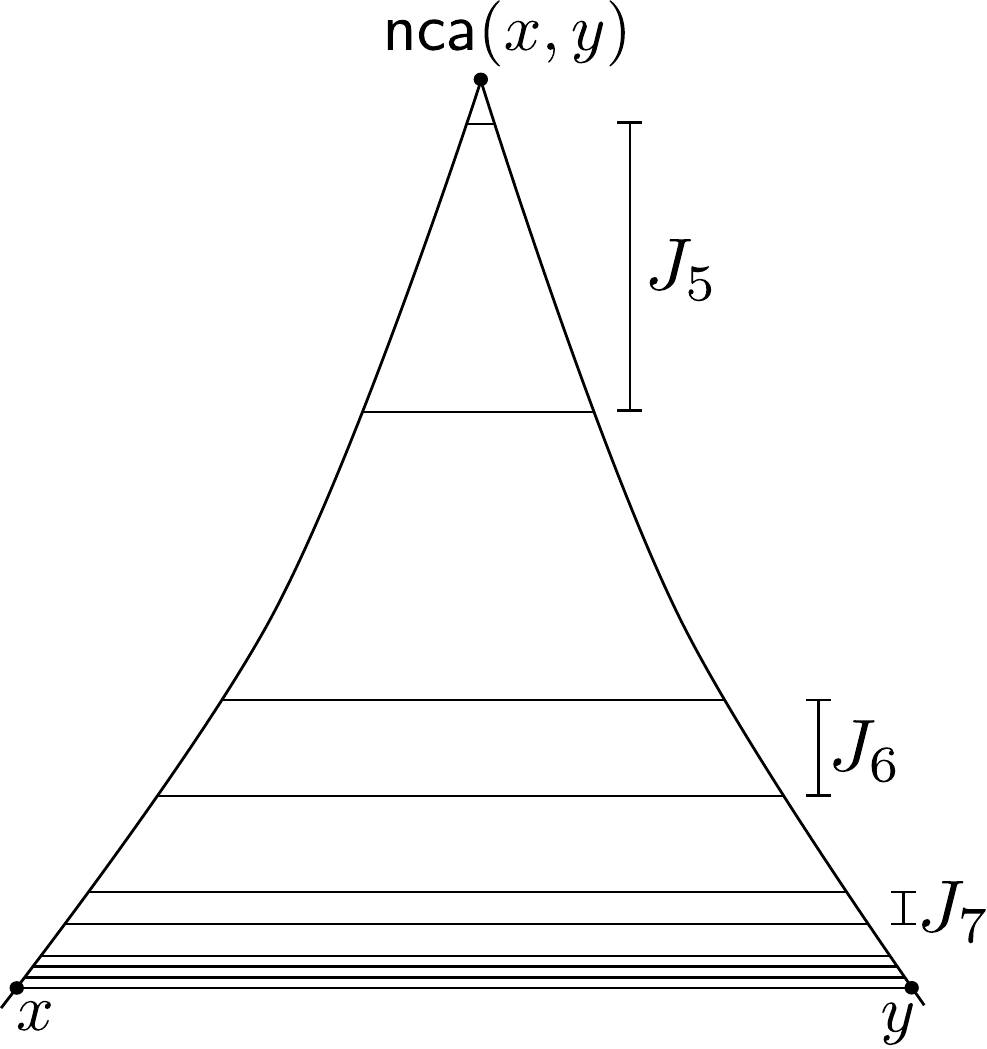}}
\end{center}

\begin{lemma}\label{lem:if_bend_then_cut}
Let $i\in \{1,\ldots,\delta\}$, $j\in J_i$, such that $S_i$ is $j$-straight.
Let $Y \subseteq E(G')$, such that $S_{i+1}$ is the $(\Delta / 3^i,Y)$-shift of $S_i$, where $|Y\cap E(P)|$ is odd.
Then, $\mathbf{1}_{S_{\delta}}(x) \neq \mathbf{1}_{S_{\delta}}(y)$.
\end{lemma}
\begin{proof}
Since $j\in J_i$, it follows that $j \leq h-\Delta/3^i$.
Let $x'$ be the ancestor of $x$ with $d_G(x,x') = \Delta / 3^i$,
and let $y'$ be the ancestor of $y$ with $d_G(y,y') = \Delta / 3^i$.
Since $S_i$ is $j$-straight, it follows that $x',y'\in S_i$.
Suppose that $S_{i+1}$ is an odd shift of $S_i$; the case where $S_{i+1}$ is an even shift of $S_i$ is completely symmetric by exchanging $x$, and $y$.
Since $S_{i+1}$ is the odd $(\Delta / 3^i, Y)$-shift of $S_i$, and $|Y\cap E(P)|$ is odd, it follows that $x\in S_{i+1}$, and $y \notin S_{i+1}$.
Since $S_{\delta} \supset \ldots \supset S_{i+1}$, we obtain $x\in S_{\delta}$.
Let $W$ be the subpath of the ray containing $y$, and $y'$.
We have that for any $k>1$, the cut $S_{i+k}$ contains a prefix of $W$ of length at most $\Delta / 3^{i+1} + \ldots + \Delta / 3^{\delta} < 2 \Delta / 3^{i+1} \leq h-j = d_G(y,y')$, and therefore $y$ is not in $S_{\delta}$, concluding the proof.
\end{proof}

For any $t\in \{h',\ldots,h\}$, let
\[
E_t = \{\{z,w\} \in E(P) : \depth(\nca(z,w)) \leq t\}.
\]

\begin{lemma}\label{lem:t_star}
There exists $t^* \in \{h-L/2,\ldots,h-L/4\}$, such that $L/2\leq |E_{t^*}| \leq L$.
\end{lemma}
\begin{proof}
By lemma \ref{lem:E_top} we have $|E_{h-L/2}| \leq L$.
Let $x'$ be the ancestor of $x$ with $d_G(x,x')=L/4$, and let $y'$ be the ancestor of $y$ with $d_G(y,y')=L/4$.
We have $|E_{h-L/4}| \geq d_G(x',y') \geq d_G(x,y)-d_G(x,x')-d_G(y,y') \geq L/2$, and the assertion follows.
\end{proof}

\iffalse
Observe that for any $i\neq i'$, we have $J_i \cap J_{i'} = \emptyset$.
For any $i\in \{1,\ldots,\delta\}$, and for any $j\in J_i$, let
\[
E_j = \{\{z,w\}\in E(P) : j - \depth(\nca(z,w)) \in [\Delta/3^i, 6  \Delta/3^i)\}.
\]
We argue that
\[
|E_j| = O(L).
\]
\marginpar{*** proof needed}
\fi

For any $i\in \{1,\ldots,\delta\}$, and for any $j\in \{1,\ldots,\Delta'\}$, 
let
\[
B_{i,j} = \bigcup_{t = j-6 \Delta / 3^i}^{j- \Delta/3^i} E_t
\]
Intuitively, the set $B_{i,j}$ contains all edges in $H$ that could possibly bend $S_i$, when $S_i$ is $j$-straight.

\begin{lemma}\label{lem:i_star}
There exists $i^*\in \{1,\ldots,\delta\}$, with $\Delta / 3^{i^*} = \Omega(L)$, and such that 
\begin{align}
\bigcup_{i = 1}^{i^*-1} \bigcup_{j\in J_i} B_{i,j} = O(L), \label{eq:i_star1}
\end{align}
and
\begin{align}
\bigcup_{i = 1}^{i^*} \bigcup_{j\in J_i} B_{i,j} = \Omega(L). \label{eq:i_star2}
\end{align}
\end{lemma}
\begin{proof}
It is straightforward to verify that 
\[
\bigcup_{i = 1}^{\delta} \bigcup_{j\in J_i} B_{i,j} = E(P).
\]
Let $t^*$ be as in lemma \ref{lem:t_star}.
It follows that by setting 
\[
i^* = \min\{i\in \{1,\ldots,\delta\} : h-\Delta/3^i \leq t^*\},
\]
we have $E_{t^*} \subseteq \bigcup_{i = 1}^{i^*} \bigcup_{j\in J_i} B_{i,j}$, and therefore 
conditions \eqref{eq:i_star1} \& \eqref{eq:i_star2} are satisfied.
Moreover, we have $\Delta / 3^{i^*} \geq h-t^* \geq L/4 = \Omega(L)$, as required.
\end{proof}

\begin{lemma}\label{lem:straight_lower}
Let $i\in \{1,\ldots,i^*\}$, and $j\in J_i$.
Then, $\Pr[S_i \text{ is } j\text{-straight}] = \Omega(1/\Delta')$. 
\end{lemma}
\begin{proof}
We use a coupling argument.
Let $E^* = \bigcup_{i = 1}^{i^*-1} \bigcup_{j\in J_i} B_{i,j}$.
Our embedding uses a random process $\sigma$ of sampling $S_1,\ldots,S_i$.
We define a modified random process $\sigma'$ for sampling monotone cuts $S_1',\ldots,S_i'$ as follows.
The process $\sigma'$ uses the same algorithm as $\sigma$, with the only difference that when taking $S_{\iota+1}'$ to be the $(r,Y)$-shift of $S_{\iota}'$, where $Y$ is chosen from a set $E'=E'(S_{\iota}')$ independently with probability $1/r$, we instead chose $Y'$ from the set $E'\setminus E^*$.
In other words, we execute the same algorithm, but we always ignore the edges in $E^*$ when computing shifts.

Since we ignore the edges in $E^*$, 
the final cut $S_i'$ is always either $j'$-straight, for some $j'\leq j$, or it contains $P_j$.
Arguing as in the proof of lemma \ref{lem:uniform_vertical}, it is straightforward to verify that every vertex $u\in V(G)$ appears in $\partial_V S_i$ with probability $1/\Delta'$.
Therefore, $\Pr[S_i' \text{ is } j\text{-straight}] = 1/\Delta'$.

We can now re-define the original random process $\sigma$ as follows.
At every step, when computing a $(r,Y)$-shift of some cut $S$, we first pick $Y'$ according to the choices made in $\sigma'$, and then we augment $Y'$ to a new set $Y$ by adding independently, and with probability $1/r$ every edge from $E^*$ that was ignored in $\sigma'$ at this step.
It is straightforward to check that this definition results in the same process $\sigma$.

Let us say that the coupling of these two processes \emph{fails} if at some step the processes $\sigma$, and $\sigma'$ deviate.
Recall that we obtain a set $S_{\iota+1}$ by taking a $(r,Y)$-shift of a set $S_{\iota}$, for some $r=\Delta/3^{\iota}$.
Every edge in $E^*$ is eligible for appearing in such a set $Y$ at most $O(1)$ times during the process.
Moreover, since $\iota \leq i^*$, we have that every eligible edge is chosen in a set $Y$ with probability $1/r = O(3^{i^*}/\Delta) = O(1/L)$.
It follows that for any execution of $\sigma'$, the coupling does not fail with at least some constant probability $q>0$.
Thus, $\Pr[S_i \text{ is } j\text{-straight}] \geq q\cdot \Pr[S_i'\text{ is } j\text{-straight}] = \Omega(1/\Delta')$, as required.
\end{proof}

We will use the following simple fact about the parity of the sum of independent Bernoulli random variables.

\begin{proposition}\label{prop:odd}
There exists $c>0$, such that the following holds.
Let $p,k>0$, and let $X_1,\ldots,X_{k}$ be a collection of independent Bernoulli random variables, such that for any $i\in \{1,\ldots,k\}$, we have $\Pr[X_i=1] = p$.
Then, $\Pr\left[\sum_{i=1}^{k} X_i \text{ is odd} \right] > \min\{1/4, cpk\}$.
\end{proposition}

\begin{lemma}[Contraction of horizontal pairs]\label{lem:horizontal_contraction}
Let $x,y\in V(G)$, such that $\depth(x)=\depth(y)$.
Then, $\|f(x)-f(y)\|_1 = \Omega( d(x,y) / \Delta')$.
\end{lemma}
\begin{proof}
For any $i\in \{1,\ldots,i^*\}$, and for any $j\in J_i$, let ${\cal W}_1(i,j)$ denote the event that $S_i$ is $j$-straight.
By lemma \ref{lem:straight_lower} we have
\[
\Pr[{\cal W}_1(i,j)] = \Omega(1/\Delta').
\]
Let ${\cal W}_2(i,j)$ denote the event that there exists $Y \subseteq E(G')$, such that $|Y \cap E(P)|$ is odd, and $S_{i+1}$ is the $(\Delta / 3^i,Y)$-shift of $S_i$.
Conditioned on the event that $S_i$ is $j$-straight, we have that $S_{i+1}$ is the $(\Delta / 3^i, Y)$-shift of $S_i$, for some random $Y\subset E(G')$, with $Y\cap E(P) \subseteq B_{i,j}$, where every element of $B_{i,j}$ is chosen independently with probability $p=3^i/\Delta$.
Applying Proposition \ref{prop:odd} we deduce that
\[
\Pr[{\cal W}_2(i,j) | {\cal W}_1(i,j)] = \Omega(\min\{1/4,|B_{i,j}| 3^i / \Delta\}) = \Omega(\min\{1/4,|E_j| / (h-j)\}).
\]

Consider some $e=\{z,w\}\in B_{i,j}$, with $\depth(\nca(z,w))=h''$.
The edge $e$ appears in $B_{i',j'}$, for at least $\Omega(h-h'')$ different values of $j'\in \bigcup_{i'=1}^{i^*} J_{i'}$.
Arguing as in the proof of lemma \ref{lem:horizontal_expansion}, we can show that for every such value $j'$, we have $h-j' = \Theta(h-h'')$.
Therefore,
\[
\sum_{i\in\{1,\ldots,i^*\}} \sum_{j \in J_i} |B_{i,j}| / (h-j) = \Omega\left( \left| \bigcup_{i\in\{1,\ldots,i^*\}} \bigcup_{j \in J_i} B_{i,j} \right|   \right) = \Omega(L)
\]

Combining the above with lemma \ref{lem:if_bend_then_cut}, we obtain
\begin{align*}
\|f(x) - f(y)\|_1 &= \Pr[\mathbf{1}_{S_{\delta}}(x) \neq \mathbf{1}_{S_{\delta}}(x)]\\
 &\geq \sum_{i\in\{1,\ldots,i^*\}} \sum_{j \in J_i} \Pr[{\cal W}_2(i,j) | {\cal W}_1(i,j)] \cdot \Pr[{\cal W}_1(i,j)]\\
 &= \Omega(1/\Delta') \min\left\{ 1/4, \sum_{i\in\{1,\ldots,\delta\}} \sum_{j \in J_i} |B_{i,j}| / (h-j) \right\}\\
 &= \Omega(L/\Delta')\\
 &= \Omega(d_G(x,y)/\Delta'),
\end{align*}
as required.
\end{proof}

\subsection{Distortion of general pairs of points}

\begin{lemma}[Embedding pyramids into $L_1$]\label{lem:pyramids_L_1}
There exists a universal constant $c > 1$, such that 
every pyramid graph admits an embedding into $L_1$ with distortion at most $c$.
\end{lemma}
\begin{proof}
We will show that the embedding $g = f \oplus f_0$ has constant distortion on $G$.
Let $x,y\in V(G)$.
Assume w.l.o.g.~that $\depth(x) \geq \depth(y)$.
Let $R_x$ be the ray containing $x$, and let $x'$ be the unique vertex in $R_x$, with $\depth(x')=\depth(y)$.
By lemmas \ref{lem:horizontal_expansion} \& \ref{lem:horizontal_contraction} we have that there exist universal constants $\alpha > \beta > 0$, such that for any
\begin{align}
\beta d_G(x',y)/\Delta &\leq \|f(x')-f(y)\|_1 \leq \alpha d_G(x',y)/\Delta \label{eq:gen1}
\end{align}
Note that 
\begin{align}
d_G(x,x')=\depth(x)-\depth(x') = \depth(x)-\depth(y) \geq d_G(x,y). \label{eq:gen2}
\end{align}
Thus, we have
\begin{align}
\|f(x) - f(y)\|_1 &\leq \|f(x)-f(x')\|_1 + \|f(x')-f(y)\|_1 \label{eq:gen_exp1} \\
 &\leq d_G(x,y)/\Delta + \alpha d_G(x',y)/\Delta \label{eq:gen_exp2} \\
 &\leq d_G(x,y)/\Delta + \alpha d_G(x',x)/\Delta + \alpha d_G(x,y)/\Delta \notag \\
 &= (\alpha+1) d_G(x,x')/\Delta + \alpha d_G(x,y)/\Delta \notag \\
 &\leq (\alpha+1) d_G(x,y)/\Delta + \alpha d_G(x,y)/\Delta \label{eq:gen_exp3} \\
 &= (2\alpha+1) d_G(x,y)/\Delta, \label{eq:f_general_expansion}
\end{align}
where \eqref{eq:gen_exp1} follows by the triangle inequality, \eqref{eq:gen_exp2} by lemma \ref{lem:vertical_pairs} \& \eqref{eq:gen1}, and \eqref{eq:gen_exp3} by \eqref{eq:gen2}.
By \eqref{eq:f_general_expansion} we have
\begin{align}
\|g(x) - g(y)\|_1 &= \|f(x)-f(y)\|_1 + \|f_0(x)-f_0(y)\|_1 \notag \\
 &\leq (2\alpha + 1) d_G(x,y)/\Delta + d_G(x,y)/\Delta \notag \\
 &= (2\alpha + 2) d_G(x,y)/\Delta. \label{eq:g_expansion}
\end{align}
This bounds the expansion of $g$.
It remains to bound the contraction of $g$.

Let $\gamma = \frac{\beta}{4(2\alpha+1)}$.
Assume first that $d_G(x,y) \geq \gamma d_G(x,y)$.
We have 
\begin{align}
\|g(x)-g(y)\|_1 &\geq \|f_0(x) - f_0(y)\|_1 \notag \\
 &= d_G(x',y)/\Delta \notag \\
 &\geq \gamma d_G(x,y)/\Delta \label{eq:exp_g_1}
\end{align}
Next, assume that $d_G(x,y) < \gamma d_G(x,y)$.
We have
\begin{align}
\|g(x)-g(y)\|_1 &\geq \|f(x) - f(y)\|_1 \notag \\
 &\geq \|f(x') - f(y)\|_1 - \|f(x) - f(x')\|_1 \notag \\
 &\geq \beta d_G(x',y)/\Delta - (2\alpha+1) d_G(x,x') / \Delta \label{eq:ggg1}\\
 &> (1-\gamma)\beta d_G(x,y)/\Delta - \gamma (2\alpha+1) d_G(x,y) / \Delta \notag\\
 &> \frac{1}{2} \beta d_G(x,y) / \Delta \label{eq:ggg2}
\end{align}
where \eqref{eq:ggg1} follows by \eqref{eq:gen1} \& \eqref{eq:f_general_expansion}.
Combining \eqref{eq:exp_g_1} \& \eqref{eq:ggg2}, we obtain that for all $x,y\in V(G)$
\begin{align}
\|g(x) - g(y)\|_1 &\geq \frac{\beta}{4(2\alpha+1)} d_G(x,y) / \Delta. \label{eq:ggg3}
\end{align}
From \eqref{eq:g_expansion} \& \eqref{eq:ggg3} we conclude that the distortion of $g$ is at most $4(2\alpha+1)(2\alpha+2)/\beta = O(1)$, concluding the proof.
\end{proof}

\subsection{Proof of the main result}

Combining the above results, we can now prove our main theorem.

\begin{proof}{Proof of theorem \ref{thm:main}}
Let $(X,d)$ be a planar metric of non-positive curvature.
Using lemma \ref{lem:funnels}, the metric $(X,d)$ admits an  embedding into some funnel $G$ with distortion $c_1=O(1)$.
Using lemma \ref{lem:pyramids} we can find a stochastic embedding of $G$ into a distribution ${\cal F}$ over pyramids with distortion $c_2=O(1)$.
By lemma \ref{lem:pyramids_L_1} every pyramid in the support of ${\cal F}$ admits an embedding into $L_1$ with distortion $c_3=O(1)$.
Combining with lemma \ref{lem:stochastic_L_1} we obtain that $G$ admits an embedding into $L_1$ with distortion $c_2c_3$.
Therefore $(X,d)$ admits an embedding into $L_1$ with distortion $\gamma=c_1c_2c_3 = O(1)$, concluding the proof.
\end{proof}

\iffalse
\section{A hyperbolization conjecture}

\begin{conjecture}[Hyperbolization]\label{conjecture:hyper}
Every planar graph metric can be embedded with constant distortion into a convex combination of hyperbolic planar metrics.
\end{conjecture}

We remark that conjecture \ref{conjecture:hyper} is \emph{equivalent} to the standard planar embedding conjecture.
\fi

\paragraph{Acknowledgements}
The author thanks James R.~Lee for sharing numerous insights on the planar embedding conjecture, and for explaining some of the previous work on embeddings of hyperbolic spaces.

\bibliography{bibfile,pathwidth,trees,treewidth}

\newcommand{\etalchar}[1]{$^{#1}$}
\def\cprime{$'$} \def\cprime{$'$}
\begin{thebibliography}{CGN{\etalchar{+}}06}

\bibitem[AR98]{AR98}
Yonatan Aumann and Yuval Rabani.
\newblock An {$O(\log k)$} approximate min-cut max-flow theorem and
  approximation algorithm.
\newblock {\em SIAM J. Comput.}, 27(1):291--301 (electronic), 1998.

\bibitem[BS]{Bonk_embeddingsof}
Mario Bonk and Oded Schramm.
\newblock Embeddings of gromov hyperbolic spaces.
\newblock {\em Geom. Funct. Anal}, 10:266--306.

\bibitem[BS05a]{Buyalo:2005to}
Sergei Buyalo and Viktor Schroeder.
\newblock {A product of trees as universal space for hyperbolic groups}.
\newblock {\em arXiv.org}, September 2005.

\bibitem[BS05b]{hyperbolic_into_products_of_trees}
Sergei Buyalo and Viktor Schroeder.
\newblock Embedding of hyperbolic spaces in the product of trees.
\newblock {\em Geometriae Dedicata}, 113(1):75--93, 2005.

\bibitem[CGN{\etalchar{+}}06]{CGNRS06}
Chandra Chekuri, Anupam Gupta, Ilan Newman, Yuri Rabinovich, and Alistair
  Sinclair.
\newblock Embedding {$k$}-outerplanar graphs into {$l\sb 1$}.
\newblock {\em SIAM J. Discrete Math.}, 20(1):119--136 (electronic), 2006.

\bibitem[CJLV08]{DBLP:conf/focs/ChakrabartiJLV08}
Amit Chakrabarti, Alexander Jaffe, James~R. Lee, and Justin Vincent.
\newblock Embeddings of topological graphs: Lossy invariants, linearization,
  and 2-sums.
\newblock In {\em FOCS}, pages 761--770, 2008.

\bibitem[CSW13]{DBLP:journals/jct/ChekuriSW13}
Chandra Chekuri, F.~Bruce Shepherd, and Christophe Weibel.
\newblock Flow-cut gaps for integer and fractional multiflows.
\newblock {\em J. Comb. Theory, Ser. B}, 103(2):248--273, 2013.

\bibitem[DL97]{GeoCuts}
Michel~Marie Deza and Monique Laurent.
\newblock {\em Geometry of cuts and metrics}, volume~15 of {\em Algorithms and
  Combinatorics}.
\newblock Springer-Verlag, Berlin, 1997.

\bibitem[GNRS04]{GNRS99}
Anupam Gupta, Ilan Newman, Yuri Rabinovich, and Alistair Sinclair.
\newblock Cuts, trees and {$l\sb 1$}-embeddings of graphs.
\newblock {\em Combinatorica}, 24(2):233--269, 2004.

\bibitem[KL06]{DBLP:conf/focs/KrauthgamerL06}
Robert Krauthgamer and James~R. Lee.
\newblock Algorithms on negatively curved spaces.
\newblock In {\em FOCS}, pages 119--132, 2006.

\bibitem[KPR93]{DBLP:conf/stoc/KleinPR93}
Philip~N. Klein, Serge~A. Plotkin, and Satish Rao.
\newblock Excluded minors, network decomposition, and multicommodity flow.
\newblock In {\em STOC}, pages 682--690, 1993.

\bibitem[LLR95]{LLR95}
N.~Linial, E.~London, and Y.~Rabinovich.
\newblock The geometry of graphs and some of its algorithmic applications.
\newblock {\em Combinatorica}, 15(2):215--245, 1995.

\bibitem[LP13]{2sums}
James~R. Lee and Daniel Poore.
\newblock On th e 2-sum embedding conjecture.
\newblock In {\em SoCG}, 2013.

\bibitem[LR99]{LR99}
Tom Leighton and Satish Rao.
\newblock Multicommodity max-flow min-cut theorems and their use in designing
  approximation algorithms.
\newblock {\em J. ACM}, 46(6):787--832, 1999.

\bibitem[LR10]{DBLP:journals/dcg/LeeR10}
James~R. Lee and Prasad Raghavendra.
\newblock Coarse differentiation and multi-flows in planar graphs.
\newblock {\em Discrete {\&} Computational Geometry}, 43(2):346--362, 2010.

\bibitem[LS09]{DBLP:conf/stoc/LeeS09}
James~R. Lee and Anastasios Sidiropoulos.
\newblock On the geometry of graphs with a forbidden minor.
\newblock In {\em STOC}, pages 245--254, 2009.

\bibitem[OS81]{Okamura198175}
Haruko Okamura and P.D. Seymour.
\newblock Multicommodity flows in planar graphs.
\newblock {\em Journal of Combinatorial Theory, Series B}, 31(1):75 -- 81,
  1981.

\bibitem[Pap05]{papadopoulos2005metric}
A.~Papadopoulos.
\newblock {\em Metric spaces, convexity and nonpositive curvature}.
\newblock Irma Lectures in Mathematics and Theoretical Physics Series. European
  Mathematical Society Publishing House, 2005.

\bibitem[Rao99]{DBLP:conf/compgeom/Rao99}
Satish Rao.
\newblock Small distortion and volume preserving embeddings for planar and
  euclidean metrics.
\newblock In {\em Symposium on Computational Geometry}, pages 300--306, 1999.

\bibitem[TL97]{thurston1997three}
W.P. Thurston and S.~Levy.
\newblock {\em Three-dimensional geometry and topology. 1 (1997)}.
\newblock Princeton mathematical series. Princeton University Press, 1997.

\end{thebibliography}
\bibliographystyle{alpha}

\end{document}